\documentclass[12pt, article, onecolumn]{IEEEtran}
\IEEEoverridecommandlockouts                              % This command is only                                 
\usepackage{graphics} % for pdf, bitmapped graphics files
\usepackage{epsfig} % for postscript graphics files
\usepackage{subfigure}
\usepackage{times} % assumes new font selection scheme installed
\usepackage{amsmath} % assumes amsmath package installed
\usepackage{amssymb}  % assumes amsmath package installed
\usepackage{mathrsfs}
\usepackage{cite}
\usepackage{epsf}

\usepackage{theorem}
\usepackage[active]{srcltx}
\usepackage{hyperref}
\usepackage{cleveref}

\newtheorem{theorem}{Theorem}
\newtheorem{proposition}{Proposition}

\newtheorem{corollary}{Corollary}
\newtheorem{remark}{Remark}
\newtheorem{definition}{Definition}
\newtheorem{problem}{Problem}

\newcommand{\ud}{\text{d}}
\newcommand{\Def}{\stackrel{\triangle}{=}}

\crefname{lemma}{Lemma}{Lemmas}
\Crefname{lemma}{Lemma}{Lemmas}
\crefname{theorem}{Theorem}{Theorems}
\Crefname{theorem}{Theorem}{Theorems}
\crefname{section}{Section}{Sections}
\Crefname{section}{Section}{Sections}
\crefname{figure}{Fig.}{Figs.}
\Crefname{figure}{Figure}{Figures}
\crefname{table}{Table}{Tables}
\Crefname{table}{Table}{Tables}
\crefname{equation}{Eq.}{Eqs.}
\Crefname{equation}{Equation}{Equations}
\crefname{proposition}{Proposition}{Propositions}
\Crefname{proposition}{Proposition}{Propositions}
\crefname{corollary}{Corollary}{Corollaries}
\Crefname{corollary}{Corollary}{Corollaries}
\crefname{problem}{Problem}{Problems}
\Crefname{problem}{Problem}{Problems}
\crefname{definition}{Definition}{Definitions}
\Crefname{definition}{Definition}{Definitions}

\ifCLASSINFOpdf
\else
\fi
\begin{document}

\title{Maximum Throughput in Multiple-Antenna Systems}

%\author{Mahdi~Zamani,~\IEEEmembership{Student Member,~IEEE,}
%        Amir~K.~Khandani,~\IEEEmembership{Member,~IEEE}
%%        and~Jane~Doe,~\IEEEmembership{Life~Fellow,~IEEE}
%\thanks{M. Zamani is with the Department
%of Electrical and Computer Engineering, University of Waterloo, Waterloo, ON N2L 3G1 Canada e-mail: mzamani@cst.uwaterloo.ca.}
%\thanks{A. K. Khandani is with the Department
%of Electrical and Computer Engineering, University of Waterloo, Waterloo, ON N2L 3G1 Canada e-mail: khandani@cst.uwaterloo.ca.}
%        }

%\author{\IEEEauthorblockN{Mahdi Zamani}
%\IEEEauthorblockA{Department of Electrical and\\Computer Engineering\\
%University of Waterloo\\
%Waterloo, Ontario N2L 3G1\\
%Email: mzamani@cst.uwaterloo.ca}
%%\and
%%\IEEEauthorblockN{Homer Simpson}
%%\IEEEauthorblockA{Twentieth Century Fox\\
%%Springfield, USA\\
%%Email: homer@thesimpsons.com}
%\and
%\IEEEauthorblockN{Amir K. Khandani}
%\IEEEauthorblockA{Department of Electrical and\\Computer Engineering\\
%University of Waterloo\\
%Waterloo, Ontario N2L 3G1\\
%Email: khandani@cst.uwaterloo.ca}}

%

\author{\IEEEauthorblockN{Mahdi Zamani and Amir K. Khandani}
\IEEEauthorblockA{\\ 
%Coding and Signal Transmission Lab. \\
Department of Electrical and Computer Engineering\\
University of Waterloo, 
Waterloo, ON N2L 3G1\\
Emails: \{mzamani, khandani\}@cst.uwaterloo.ca}
%\and
%\IEEEauthorblockN{Homer Simpson}
%\IEEEauthorblockA{Twentieth Century Fox\\
%\and
%\IEEEauthorblockN{Amir K. Khandani}
%\IEEEauthorblockA{Department of Electrical and\\Computer Engineering\\
%University of Waterloo\\
%Waterloo, Ontario N2L 3G1\\
%Email: khandani@cst.uwaterloo.ca}
}

\maketitle
\thispagestyle{empty}
\pagestyle{empty}

%%%%%%%%%%%%%%%%%%%%%%%%%%%%%%%%%%%%%%%%%%%%%%%%%%%%%%%%%%%%%%%%%%%%%%%%%%%%%%%%%%%%%%%%%%%%%%%%%%%%%%%%%%%%%%%%%%%
%%%%%%%%%%%%%%%%%%%%%%%%%%%%%%%%%%%%%%%%%%%%%%%%%%%%%%%%%%%%%%%%%%%%%%%%%%%%%%%%%%%%%%%%%%%%%%%%%%%%%%%%%%%%%%%%%%%
%%%%%%%%%%%%%%%%%%%%%%%%%%%%%%%%%%%%%%%%%%%%%%%%%%%%%%%%%%%%%%%%%%%%%%%%%%%%%%%%%%%%%%%%%%%%%%%%%%%%%%%%%%%%%%%%%%%
%%%%%%%%%%%%%%%%%%%%%%%%%%%%%%%%%%%%%%%%%%%%%%%%%%%%%%%%%%%%%%%%%%%%%%%%%%%%%%%%%%%%%%%%%%%%%%%%%%%%%%%%%%%%%%%%%%%
%%%%%%%%%%%%%%%%%%%%%%%%%%%%%%%%%%%%%%%%%%%%%%%%%%%%%%%%%%%%%%%%%%%%%%%%%%%%%%%%%%%%%%%%%%%%%%%%%%%%%%%%%%%%%%%%%%%

\begin{abstract}

%Multiple antennas improve the ergodic (Shannon) capacity, outage capacity, and throughput of wireless systems. 
The point-to-point multiple-antenna channel is investigated in uncorrelated block fading environment with Rayleigh distribution. The maximum throughput and maximum expected-rate of this channel are derived under the assumption that the transmitter is oblivious to the channel state information (CSI), however, the receiver has perfect CSI.
%In this paper, the maximum throughput and maximum expected-rate of multiple-antenna systems with Rayleigh block fading links has been derived. We assume that the transmitter is oblivious to the channel state information (CSI),  however, the receiver has perfect information about it. 
First, we prove that in multiple-input single-output (MISO) channels, the optimum transmission strategy maximizing the throughput is to use all available antennas and perform equal power allocation with uncorrelated signals. 
Furthermore, to increase the expected-rate, multi-layer coding is applied. Analogously, we establish that sending uncorrelated signals and performing equal power allocation across all available antennas at each layer is optimum. 
A closed form expression for the maximum continuous-layer expected-rate of MISO channels is also obtained.
%The maximum continuous-layer expected-rate of the MISO channel is obtained in a closed form expression by sending a continuum of code layers at the transmitter.
Moreover, we investigate multiple-input multiple-output (MIMO) channels, and formulate the maximum throughput in the asymptotically low and high SNR regimes and also asymptotically large number of transmit or receive antennas by obtaining the optimum transmit covariance matrix. 
%It is also shown that in MISO channels as well as asymptotically low and high SNR regimes MIMO channels, channel correlation (at the transmitter or receiver) always reduces the maximum throughput.
%The maximum throughput formula of such cases are presented. 
Finally, a distributed antenna system, wherein two single-antenna transmitters want to transmit a common message to a single-antenna receiver, is considered. It is shown that this system has the same outage probability and hence, throughput and expected-rate, as a point-to-point $2\times 1$ MISO channel.

%Finally, it is proved that distributed antenna systems with two separate transmitters and one receiver achieve
%all achievable instantaneous mutual information distribution of the MISO channel with two transmitters. Thus, the outage capacity and the maximum throughput and expected-rate of both systems are the same.
%% the outage capacity and the maximum throughput and expected-rate of MISO channels with two transmit antennas.
%% in terms of the instantaneous mutual information distribution. 
%In the other word, power cooperation and full cooperation between the transmitters does not increase the throughput and expected-rate of two-transmitter distributed antenna systems. 
%%All theoretical results are illustrated by numerical simulations. 

\end{abstract}
\IEEEpeerreviewmaketitle

\let\thefootnote\relax\footnotetext{Financial supports provided by Natural Sciences and Engineering Research Council of Canada (NSERC) and Ontario Ministry of Research \& Innovation (ORF-RE) are gratefully acknowledged.}

%%%%%%%%%%%%%%%%%%%%%%%%%%%%%%%%%%%%%%%%%%%%%%%%%%%%%%%%%%%%%%%%%%%%%%%%%%%%%%%%%%%%%%%%%%%%%%%%%%%%%%%%%%%%%%%%%%%
%%%%%%%%%%%%%%%%%%%%%%%%%%%%%%%%%%%%%%%%%%%%%%%%%%%%%%%%%%%%%%%%%%%%%%%%%%%%%%%%%%%%%%%%%%%%%%%%%%%%%%%%%%%%%%%%%%%
%%%%%%%%%%%%%%%%%%%%%%%%%%%%%%%%%%%%%%%%%%%%%%%%%%%%%%%%%%%%%%%%%%%%%%%%%%%%%%%%%%%%%%%%%%%%%%%%%%%%%%%%%%%%%%%%%%%
%%%%%%%%%%%%%%%%%%%%%%%%%%%%%%%%%%%%%%%%%%%%%%%%%%%%%%%%%%%%%%%%%%%%%%%%%%%%%%%%%%%%%%%%%%%%%%%%%%%%%%%%%%%%%%%%%%%
%%%%%%%%%%%%%%%%%%%%%%%%%%%%%%%%%%%%%%%%%%%%%%%%%%%%%%%%%%%%%%%%%%%%%%%%%%%%%%%%%%%%%%%%%%%%%%%%%%%%%%%%%%%%%%%%%%%

\section{Introduction}
The information theoretic aspects of wireless fading channels have received wide attention \cite{biglieri}.   
The growing demand for QoS and network coverage inspires the use of multiple-antenna arrays at the transmitter and/or receiver \cite{foschini1996layered,tarokh1998space,tarokh1999space,wolniansky1998v}. It has been shown that multiple-antenna arrays have the ability to reach higher transmission rates \cite{telatar , marzetta,foschini1998limits}.
With no delay constraint, the ergodic nature of the fading channel can be experienced by sending very large transmission blocks, and the ergodic capacity is well studied \cite{biglieri}. 
When the channel variation is slow, the channel can be estimated relatively accurately at the receiver. By assuming perfect CSI at the receiver but no CSI at the transmitter, Telatar \cite{telatar} showed that the ergodic capacity of general MIMO channels is achieved by sending an uncorrelated circularly symmetric zero mean equal power complex Gaussian codebook on all transmit antennas.

Due to the stringent delay constraint for the problem in consideration, the transmission block length is forced to be shorter than the dynamics of the slow fading process, though still large enough to yield a reliable communication. The performance of such channels are usually evaluated by outage capacity. The notion of capacity versus outage was introduced in \cite{biglieri,ozarow}.
Jorswieck and Boch \cite{jorswieck} proved that in uncorrelated MISO channels, the optimum transmit strategy minimizing the outage probability is to use a fraction of all available transmit antennas and perform equal power allocation with uncorrelated signals.% (see \cite{jorswieck} theorem 3 in section 3.1). 

%The performance metric of block fading channels occurs also in the maximum throughput expression \cite{jorswieck}. 
The maximum throughput is an important performance measure in block fading channels \cite{ahmed2003throughput}, which is defined as the maximum of the product of the transmission rate and the probability of successful transmission using a single-layer code (see \cref{def-throughput}). 
As mentioned in \cite{jorswieck}, their results on the outage probability cannot be directly applied to this metric due to the maximization.
In this paper, we prove that to achieve the maximum throughput in an uncorrelated MISO channel, the optimum transmit strategy is to send equal power uncorrelated signals from all available antennas (see \cref{miso-single-theorem}).

%Some applications interest in the maximum average achievable rate at the receiver.  
The maximum average achievable rate is another performance measure which is important in some applications.
A good example for such applications is a TV broadcasting system where users with better channels can receive additional services such as high definition TV signals \cite{reza}. 
%The growing adoption of broadcast mobile TV services suggests that it has the potential to become a mass market application. 
Due to the large number of users, the transmitter cannot access the CSI.
In order to increase the average achievable rate, Shamai and Steiner \cite{shamai} proposed a broadcast approach (multi-layer coding) for a point-to-point block fading channel with no CSI at the transmitter. 
%As the more code layers, the higher average achievable rate, 
Since the average achievable rate increases with the number of code layers,
they reached the highest average achievable rate using a continuous-layer (infinite-layer) code. 
This idea was applied to a two-hop single-relay channel in \cite{steiner,vahid}, a channel with two collocated cooperative users in \cite{steiner2}, and a two-hop parallel-relay network (the diamond channel) in \cite{mine2011isit}. Multi-layer coding can also achieve the maximum average achievable rate in a block fading multiple-access channel with no CSI at the transmitters \cite{tse}. The optimized trade-off between the QoS and network coverage in a multicast network was derived in \cite{reza} using the broadcast approach.
Here, we derive the maximum expected-rate of MISO channels, which is defined as the maximum average decodable rate when a multi-layer code is transmitted (see \cref{def-expected-rate}).
\Cref{miso-multi-theorem} proves that to maximize the expected-rate in MISO channels, it is optimum to transmit equal power independent signals on all available antennas in each layer. 
Using the continuous-layer coding approach, the maximum expected-rate of MISO channels is then obtained and formulated in closed form in \cref{miso-infinite-proposition}.  

To evaluate the maximum throughput in uncorrelated MIMO channels, the distribution of the instantaneous mutual information is crucial. In \cite{wang,hochwald}, it is shown that the distribution of the instantaneous mutual information in MIMO channels is always very close to the 
%always converges to a
Gaussian distribution.%, even with two transmit and receive antennas. 
%The asymptotic results for 
The mean and variance of this equivalent Gaussian distribution were derived in \cite{hochwald} for asymptotic ranges of the number of antennas.
%This paper extends the results of uncorrelated MISO channels to uncorrelated MIMO channels in asymptotically low SNR regime, asymptotically high SNR regime, asymptotically large number of transmit antennas, and asymptotically large number of receive antennas.
As this distribution is not tractable in general MIMO channels, here we consider four asymptotic cases: asymptotically low SNR regime, asymptotically high SNR regime, asymptotically large number of transmit antennas, and asymptotically large number of receive antennas. In all four cases, the optimum covariance matrix is obtained and the maximum throughput expression is derived.

Finally, the maximum throughput and maximum expected-rate of a distributed antenna system with two single-antenna transmitters and one single-antenna receiver is obtained. It is also proved that any achievable throughput, expected-rate, ergodic capacity, and outage capacity in a MISO channel with two transmit antennas are also achievable in this channel.

The rest of this paper is organized as follows. 
In \cref{prelimeneries}, the preliminaries are presented.
The maximum throughput and the maximum expected-rate of MISO channels are derived in \cref{miso-single-section,miso-multi-section}, respectively.
The maximum throughputs in four asymptotic cases of MIMO channels are obtained in \cref{mimo-section}. In \cref{distributed-section}, a distributed antenna system with two transmitters is analyzed. 
Finally, \cref{conclusion-section} concludes the paper.

%%%%%%%%%%%%%%%%%%%%%%%%%%%%%%%%%%%%%%%%%%%%%%%%%%%%%%%%%%%%%%%%%%%%%%%%%%%%%%%%%%%%%%%%%%%%%%%%%%%%%%%%%%%%%%%%%%%
%%%%%%%%%%%%%%%%%%%%%%%%%%%%%%%%%%%%%%%%%%%%%%%%%%%%%%%%%%%%%%%%%%%%%%%%%%%%%%%%%%%%%%%%%%%%%%%%%%%%%%%%%%%%%%%%%%%
%%%%%%%%%%%%%%%%%%%%%%%%%%%%%%%%%%%%%%%%%%%%%%%%%%%%%%%%%%%%%%%%%%%%%%%%%%%%%%%%%%%%%%%%%%%%%%%%%%%%%%%%%%%%%%%%%%%
%%%%%%%%%%%%%%%%%%%%%%%%%%%%%%%%%%%%%%%%%%%%%%%%%%%%%%%%%%%%%%%%%%%%%%%%%%%%%%%%%%%%%%%%%%%%%%%%%%%%%%%%%%%%%%%%%%%
%%%%%%%%%%%%%%%%%%%%%%%%%%%%%%%%%%%%%%%%%%%%%%%%%%%%%%%%%%%%%%%%%%%%%%%%%%%%%%%%%%%%%%%%%%%%%%%%%%%%%%%%%%%%%%%%%%%

\section{Preliminaries}   \label{prelimeneries}

%%%%%%%%%%%%%%%%%%%%%%%%%%%%%%

\subsection{Notation}

Throughout the paper, we represent the probability of event $A$ by $\Pr\{A\}$, and the expected and variance operations by $\mathbb{E}(\cdot)$ and $\text{Var}(\cdot)$, respectively. The notation ``$\ln$'' is used for natural logarithm, and rates are expressed in \emph{nats}. We denote $f_{\mathrm x}(\cdot)$ and $F_{\mathrm x}(\cdot)$ as the probability density function (PDF) and the cumulative density function (CDF) of random variable $\mathrm x$, respectively. For any function $F(x)$, let us define $\overline F(x) \Def 1-F(x)$ and $F^{\prime}(x) \Def \frac{\ud F(x)}{\ud x}$. $\vec{X}$ is a vector, $\mathbf{Q}$ is a matrix, and $\text{tr}(\mathbf{Q})$ denotes the trace of $\mathbf{Q}$. $\mathbf{I}_{n_t}$ denotes the $n_t \times n_t$ identity matrix. $s^o$ is the optimum solution with respect to the variable $s$. We denote the conjugation, matrix transpose, and matrix conjugate transpose operators by $^*$, $^\text{T}$, and $^\dag$, respectively. $\Re(\cdot)$ and $\Im(\cdot)$ represent the real and imaginary parts of complex variables and $|\cdot|$ represents the absolute value or modulus operator. ``$\det$'' is used for the determinant operator and $\text{eig}_{\ell} (\mathbf Q)$ is the $\ell$'th ordered eigenvalue of matrix $\mathbf Q$. Let $h_{\ell}$ denote the $\ell$'th component of vector $\vec h$, and $h_{\ell,k}$ denote the $(\ell,k)$'th entry of matrix $\mathbf H$.
$\mathcal{CN}(0,1)$ denotes the complex circularly symmetric Gaussian distribution with zero mean and unit variance and $\mathcal{N}(\mu,\sigma^2)$ denotes the Gaussian distribution with mean $\mu$ and variance $\sigma^2$. $\mathcal W_0(\cdot)$ is the zero branch of the Lambert $W$-function, also called the omega
function, which is the inverse function of $f(W)=We^W$ \cite{corless}. $\text E_1(x)$ is the exponential integral function, which is $\int_x^\infty \frac{e^{-t}}{t} \ud t,~ x \geq 0$. $\Gamma(n,x)\Def \int_x^\infty t^{n-1} e^{-t} \ud t$ is the upper incomplete gamma function, and $\Gamma(n) \Def \Gamma(n,0)$. $\digamma(n) \Def \frac{\Gamma^{\prime}(n)}{\Gamma(n)}$ and $\mathcal Q(x) \Def \frac{1}{\sqrt{2\pi}}\int_x^{\infty} e^{-\frac{t^2}{2}}\ud t$ represent the E$\ddot{\text{u}}$ler's digamma function \cite{gradshtein2000table} and $\mathcal Q$-function, respectively. 
%Throughout the paper, $\mathcal R$ represents the throughput and expected-rate, $\mathcal R^m$ represents the maximum throughput and expected-rate, and $\mathbb E \left( |X_i|^2 \right)=1, \forall i$.

%%%%%%%%%%%%%%%%%%%%%%%%%%%%%%

\subsection{Problem Setup} \label{problem-setup}

A MIMO channel with $n_t$ transmit antennas and $n_r$ receive antennas is defined as a channel with the following input-output relationship:
%, the received signal, namely $\vec Y$, is expressed as
\begin{align}
\vec Y = \mathbf H \vec X + \vec Z,
\end{align}
where $\vec Y$ is the received signal, $\mathbf H \sim \left[\mathcal{CN} (0,1)\right]_{n_r\times n_t}$ is the channel matrix, $\vec Z \sim \left[\mathcal{CN}(0,1)\right]_{n_r \times 1}$ is the independent and identically distributed (i.i.d.) additive white Gaussian noise (AWGN), and $\vec X$ is the transmitted signal under the following total power constraint:
\begin{align}
\mathbb E\left( \vec X^\dag \vec X \right)=\mathbb E\left(\text{tr}\left( \vec X \vec X^\dag \right)\right)=\text{tr}\left( \mathbb E \left( \vec X \vec X^\dag \right)\right) \leq P.
\end{align}
Defining $\mathbf Q$ as the transmit covariance matrix, i.e., $\mathbf Q=\mathbb E\left( \vec X \vec X^\dag \right)$, the instantaneous mutual information is
\begin{align} \label{instantaneous-mutual-information-mimo-general}
\mathcal I = \ln \det \left( \mathbf I_{n_r} + \mathbf H \mathbf Q \mathbf H^\dag \right) = \ln \det \left( \mathbf I_{n_t} + \mathbf Q \mathbf H^\dag \mathbf H \right).
\end{align}
%We also assume the total power constraint $P$ is assumed at the transmitter.
In a MISO channel, the channel coefficients are represented by a vector
%a vector represents the channel coefficients, namely 
$\vec h^T \sim \left[\mathcal{CN}(0,1)\right]_{n_t\times 1}$, and 
\begin{align}
Y = \vec h \vec X + Z.
\end{align}

In the following, the performance metrics which are widely used throughout the paper are defined.
\begin{definition} \label{def-throughput}
The throughput $\mathcal R_{s}$ is the average achievable rate when a single-layer code with a fixed rate $R$ is transmitted, i.e., the transmission rate times the probability of successful transmission.
The maximum throughput, namely $\mathcal R_{s}^m$, is the maximum of the throughput 
over all transmit covariance matrices $\mathbf Q$, and transmission rates $R$. Mathematically,
\begin{align} \label{throughput-definition-formula}
\mathcal R_{s}^m  \Def  \max_{
\begin{subarray}{c}
R ,\mathbf Q \\
\emph{tr}(\mathbf Q) \leq P
\end{subarray}}
 \Pr \left\{\mathcal I \geq R \right\} R.
\end{align}
\end{definition}
\begin{definition} \label{def-expected-rate}
The expected-rate $\mathcal R_{f}$ is the average achievable rate when a multi-layer code is transmitted, i.e., the statistical expectation of the achievable rate. 
The maximum expected-rate, namely $\mathcal R_{f}^m$, is the maximum of the expected-rate over all transmit covariance matrices and transmission rates in each layer, and all power distributions of the layers. Mathematically,
\begin{align} \label{finite-layer-expected-rate-definition-formula}
\mathcal R_{f}^m  \Def  \max_{
\begin{subarray}{c}
R_i, P_i ,\mathbf Q_i \\
\emph{tr}(\mathbf Q_i) \leq P_i \\
\sum_{i=1}^K P_i = P
\end{subarray}}
\sum_{i=1}^K \Pr \left\{\mathcal I_i \geq R_i \right\} R_i,
\end{align}
where $R_i$, $\mathbf Q_i$, and $\mathcal I_i$ are the transmission rate, transmit covariance matrix, and instantaneous mutual information in the $i$'th layer, respectively.
 
If a continuum of code layers are transmitted, the maximum continuous-layer (infinite-layer) expected-rate, namely $\mathcal R_{c}^m$, is given by maximizing the continuous-layer expected-rate over the layers' power distribution.
\end{definition}

\begin{definition} \label{def-ergodic-capacity}
The ergodic capacity $C_{\emph{erg}}$ is the maximum expected value of the instantaneous mutual information $\mathcal I$ over all transmit covariance matrices $\mathbf Q$. Mathematically,
\begin{align} \label{ergodic-capacity-definition-formula}
C_{\emph{erg}}  \Def  \max_{
\begin{subarray}{c}
\mathbf Q \\
\emph{tr}(\mathbf Q) \leq P
\end{subarray}}
 \mathbb E \left(\mathcal I \right).
\end{align}
\end{definition}

The main focus of this paper is to solve the following problems.
\begin{problem} \label{miso-single-prob}
To obtain the optimum transmit covariance matrix, denoted by $\mathbf Q^o$, which maximizes the throughput $\mathcal R_{s}$ in the MISO channel. 
\end{problem}

\Cref{miso-single-theorem} proves that the optimum transmit strategy is to transmit uncorrelated signals on all antennas with equal powers, i.e., $\mathbf Q^o = \frac{P}{n_t}\mathbf I_{n_t}$, and provides the maximum throughput expression.
%as $\mathcal R_{s}^m = \max_{0 \leq s \leq 1} \frac{\Gamma(n_t,n_t s)}{(n_t-1)!} \ln\left( 1+P s \right)$.
%In the proof of \cref{miso-single-theorem}, the result of \cite{jorswieck}, which proves that the optimum transmit strategy minimizing the outage probability of MISO channels is to use a fraction of antennas and perform equal power allocation, plays an important role. 

\begin{problem} \label{miso-multi-prob}
To derive the optimum transmit covariance matrix in each layer, i.e., $\mathbf Q_i^o$, for finite-layer coding in the MISO channel, which maximizes the expected-rate $\mathcal R_{f}$. 
\end{problem}

As we shall see in \cref{miso-multi-theorem}, the optimum transmit covariance matrix in each layer is in the form of $\mathbf Q_{i}^o = \frac{P_{i}}{n_t}\mathbf I_{n_t}$, and the maximum expected-rate is given by \cref{miso_avg_max_ml}.
%The maximum expected-rate is evaluated by $\mathcal R_{f}^m = \max_{\begin{subarray}{c}
%0 \leq s_l \leq 1, P_l \\
%%0 \leq P_l \leq P \\
%\sum_{l=1}^K P_l = P
%\end{subarray} }
%\sum_{l=1}^K \frac{\Gamma \left( n_t,n_t s_l\right)}{(n_t-1)!} \ln \left( 1+\frac{P_l s_l}{1+\sum_{j=l}^K P_j s_l} \right).$

\begin{problem} \label{miso-infinite-prob}
To derive the maximum continuous-layer expected-rate $\mathcal R_{c}^m$ in the MISO channel.
\end{problem}
The closed form expression of the maximum continuous-layer expected-rate is derived in \cref{miso-infinite-proposition}.
%, and the final answer is $\mathcal R_{c}^m = h(x_1) - h(x_0)$, where
%\begin{align} 
%h(x) = e^{-x}  \sum_{\ell = 1}^{n_t-1} \frac{1}{\ell !} \left( x^{\ell} - (n_t+ 1-\ell)(\ell -1)! \sum_{i = 0}^{\ell-1} \frac{x^i}{i!}  \right) \nonumber \\
%+e^{-x}-(n_t+1) \emph{E}_1 (x),
%\end{align}
%and
%$x_0$ and $x_1$ are the solutions to 
%\begin{align}
%\begin{cases}
%\sum_{\ell=0}^{n_t-1} \frac{(n_t-1)!}{\ell ! x_0^{n_t-\ell}} = 1+\frac{P}{n_t} x_0, \\
%\sum_{\ell=0}^{n_t-1} \frac{(n_t-1)!}{\ell ! x_1^{n_t-\ell}} = 1,
%\end{cases} 
%\end{align}
%respectively.

In the MIMO channel, the PDF of the instantaneous mutual information $\mathcal I$ is not known even for the simplest case of $\mathbf Q = \frac{P}{n_t}\mathbf I_{n_t}$, although there are some approximations in literature for asymptotic cases. In the next step, the maximum throughputs in four asymptotic cases of the MIMO channel are addressed. 
\begin{problem} \label{mimo-prob}
To derive the maximum throughput of the MIMO channel in asymptotically
\begin{itemize}
\item low SNR regime
\item high SNR regime
\item large number of transmit antennas
\item large number of receive antennas
\end{itemize} 
\end{problem}

Different MIMO approximations are exploited to solve \cref{mimo-prob}. For asymptotically low SNR regime, the MISO results are carried over and the maximum throughput and maximum expected-rate are formulated. For asymptotically high SNR regime, Wishart distribution properties \cite{verdu_book} are used to obtain the maximum throughput. For asymptotically large number of transmit or receive antennas, Gaussian approximations for the instantaneous mutual information presented in \cite{hochwald} are utilized. As we shall see in \cref{mimo-section}, in all aforementioned asymptotic regimes, the optimum transmit covariance matrix which maximizes the throughput is $\mathbf Q^o = \frac{P}{n_t} \mathbf I_{n_t}$.

In the last problem, a distributed antenna system consisting of two single-antenna transmitters with common messages and a single-antenna receiver is considered.  
\begin{problem} \label{distributed-prob}
To find the minimum outage probability, the maximum throughput, and the maximum expected-rate in a two-transmitter distributed antenna system. 
\end{problem}
\Cref{distributed-theorem} 
%\Cref{distributed-section} 
establishes that any achievable outage probability in the $2\times 1$ MISO channel is also achievable in the two-transmitter distributed antenna system in \cref{distributed-prob}. Hence, both channels experience the same instantaneous mutual information distribution and thereby, all MISO channel results are applied here with $n_t = 2$.

%First, let us state the following lemmas.

\subsection{A Few Useful Propositions}

In the following, we present three propositions which are used throughout the paper and they are also of independent interest.

\begin{proposition} \label{Markov-throughput-ergodic-lemma}
In fading channels, the maximum throughput is less than or equal to the ergodic capacity.
\end{proposition}

\begin{proof}
The proof is based on the Markov inequality \cite{papoulis2002probability}, that is if $f(x)=0$ for $x<0$, then, for $\alpha>0$, $\Pr \left\{x\geq\alpha\right\}\leq \frac{\mathbb E(x)}{\alpha}$. Therefore, $\forall R>0$,
\begin{align} \label{Markov-throughput-ergodic1}
\Pr\left\{ \mathcal I\geq R\right\}\leq \frac{\mathbb E\left(\mathcal I\right)}{R},
\end{align}
so that
\begin{align} \label{Markov-throughput-ergodic}
\mathcal R_{s}^m  = \max_{
\begin{subarray}{c}
R ,\mathbf Q \\
\text{tr}(\mathbf Q) \leq P
\end{subarray}}
 \Pr \left\{\mathcal I \geq R \right\} R
 \leq \max_{
\begin{subarray}{c}
\mathbf Q \\
\text{tr}(\mathbf Q) \leq P
\end{subarray}} \mathbb E\left(\mathcal I\right),
\end{align}
and \cref{Markov-throughput-ergodic} results because $\max_{
\begin{subarray}{c}
\mathbf Q, \text{tr}(\mathbf Q) \leq P
\end{subarray}} \mathbb E\left(\mathcal I\right)$ equals the ergodic capacity.
\end{proof}

\begin{proposition} \label{Markov-expected-rate-ergodic-lemma}
In fading channels, the maximum expected-rate is less than or equal to the ergodic capacity.
\end{proposition}

\begin{proof}
From \cref{finite-layer-expected-rate-definition-formula} it follows that
\begin{align} \label{Markov-expected-rate-ergodic1}
\mathcal R_{f}^m  &=  \max_{
\begin{subarray}{c}
R_i, P_i ,\mathbf Q_i \\
\text{tr}(\mathbf Q_i) \leq P_i \\
\sum_{i=1}^K P_i = P
\end{subarray}}
\sum_{i=1}^K \Pr \left\{\mathcal I_i \geq R_i \right\} R_i \nonumber \\
&\stackrel{(a)}{\leq}
\max_{
\begin{subarray}{c}
 P_i ,\mathbf Q_i \\
\text{tr}(\mathbf Q_i) \leq P_i \\
\sum_{i=1}^K P_i = P
\end{subarray}}
\sum_{i=1}^K \mathbb E\left(\mathcal I_i\right) \nonumber\\
&\stackrel{(b)}{=} \max_{
\begin{subarray}{c}
 P_i ,\mathbf Q_i \\
\text{tr}(\mathbf Q_i) \leq P_i \\
\sum_{i=1}^K P_i = P
\end{subarray}} \mathbb E\left(
\sum_{i=1}^K \mathcal I_i \right) \nonumber\\
&= \max_{
\begin{subarray}{c}
 P_i ,\mathbf Q_i \\
\text{tr}(\mathbf Q_i) \leq P_i \\
\sum_{i=1}^K P_i = P
\end{subarray}} \mathbb E\left(
\sum_{i=1}^K \ln \frac{\det \left( \mathbf I_{n_t} +\sum_{j=i}^K \mathbf Q_j \mathbf H^\dag \mathbf H \right)}{\det \left( \mathbf I_{n_t} +\sum_{j=i+1}^K \mathbf Q_j \mathbf H^\dag \mathbf H \right)}  \right) \nonumber\\
&= \max_{
\begin{subarray}{c}
 P_i ,\mathbf Q_i \\
\text{tr}(\mathbf Q_i) \leq P_i \\
\sum_{i=1}^K P_i = P
\end{subarray}} \mathbb E\left(
\ln \prod_{i=1}^K \frac{\det \left( \mathbf I_{n_t} +\sum_{j=i}^K \mathbf Q_j \mathbf H^\dag \mathbf H \right)}{\det \left( \mathbf I_{n_t} +\sum_{j=i+1}^K \mathbf Q_j \mathbf H^\dag \mathbf H \right)}  \right) \nonumber\\
&= \max_{
\begin{subarray}{c}
 P_i ,\mathbf Q_i \\
\text{tr}(\mathbf Q_i) \leq P_i \\
\sum_{i=1}^K P_i = P
\end{subarray}} \mathbb E\left(
\ln \det \left( \mathbf I_{n_t} +\sum_{i=1}^K \mathbf Q_i \mathbf H^\dag \mathbf H \right) \right),
% \nonumber\\
%&= \max_{
%\begin{subarray}{c}
% \mathbf Q \\
%\text{tr}(\mathbf Q) \leq P \\
%\end{subarray}} \mathbb E\left(
%\ln \det \left( \mathbf I_{n_t} + \mathbf Q \mathbf H^\dag \mathbf H \right) \right) \nonumber\\
%&= \max_{
%\begin{subarray}{c}
%\mathbf Q \\
%\text{tr}(\mathbf Q) \leq P
%\end{subarray}} \mathbb E\left(\mathcal I\right).
\end{align} 
where $(a)$ follows from  \cref{Markov-throughput-ergodic-lemma}, and $(b)$ follows from the fact that expectation and summation commute.
Defining $\mathbf Q \Def \sum_{i=1}^K \mathbf Q_i$, we get
\begin{align} \label{Markov-expected-rate-ergodic2}
\text{tr}\left(\mathbf Q \right) = \text{tr}\left(\sum_{i=1}^K \mathbf Q_i \right) =\sum_{i=1}^K \text{tr}\left(\mathbf Q_i \right) \leq \sum_{i=1}^K P_i = P.
\end{align}
Inserting \cref{Markov-expected-rate-ergodic2} into \cref{Markov-expected-rate-ergodic1}, we obtain
\begin{align} \label{Markov-expected-rate-ergodic}
\mathcal R_{f}^m  \leq \max_{
\begin{subarray}{c}
 \mathbf Q \\
\text{tr}(\mathbf Q) \leq P \\
\end{subarray}} \mathbb E\left(
\ln \det \left( \mathbf I_{n_t} + \mathbf Q \mathbf H^\dag \mathbf H \right) \right) 
= \max_{
\begin{subarray}{c}
\mathbf Q \\
\text{tr}(\mathbf Q) \leq P
\end{subarray}} \mathbb E\left(\mathcal I\right).
\end{align}
and \cref{Markov-expected-rate-ergodic} results because $\max_{
\begin{subarray}{c}
\mathbf Q, \text{tr}(\mathbf Q) \leq P
\end{subarray}} \mathbb E\left(\mathcal I\right)$ equals the ergodic capacity.
%and \cref{Markov-expected-rate-ergodic} results because the right-hand-side is the maximum average mutual information using a superposition code, that is less than or equal to the ergodic capacity.
\end{proof}

\Cref{Markov-throughput-ergodic-lemma,Markov-expected-rate-ergodic-lemma} lead to the fact that the maximum throughput and maximum expected-rate are upper-bounded by the ergodic capacity. 
\Cref{ergodic-capacity-miso-pro} presents the ergodic capacity of the MISO channel in closed form.

\begin{proposition} \label{ergodic-capacity-miso-pro}
The ergodic capacity in an $n_t \times 1$ MISO Rayleigh fading channel with total power constraint $P$ is given by
\begin{align} \label{ergodic-capacity-miso-formula}
C_{\emph{erg}} &= e^{\frac{n_t}{P}} \emph{E}_1 \left( \frac{n_t}{P} \right) \sum_{\ell=0}^{n_t-1} \frac{\left(-n_t \right)^\ell}{\ell!P^\ell} \nonumber\\
&+ \sum_{\ell=1}^{n_t-1} \sum_{k=0}^{\ell-1} \frac{\left(-1 \right)^k}{\left(\ell-k\right)k!} \sum_{m=0}^{\ell-k-1} \frac{\left(n_t \right)^{k+m}}{m!P^{k+m}},
\end{align}
where $\emph{E}_1 \left(\cdot\right)$ is the exponential integral function.
The ergodic capacity in a $1 \times n_r$ single-input multiple-output (SIMO) channel with total power constraint $P$ equals the ergodic capacity of an $n_r \times 1$ MISO channel with total power constraint $n_r P$.
\end{proposition}

\begin{proof}
We offer the proof in appendix \ref{appendix-ergodic-capacity-miso-pro}.
\end{proof}

\section{Maximum Throughput in MISO Channels} \label{miso-single-section}

Let the transmitted signal $\vec X$ be a single-layer code with rate $R=\ln \left( 1+P s  \right)$. In the MISO channel, the maximum throughput in \cref{throughput-definition-formula} can be rewritten as 
\begin{align} \label{miso_single1}
\mathcal R_{s}^m = \max_{\begin{subarray}{c}
R,\mathbf{Q} \\
\text{tr}(\mathbf{Q})\leq P 
\end{subarray}} \Pr \left\{ \ln \left(1+\vec h \mathbf{Q} \vec h^{\dag} \right) \geq R  \right\}R,
\end{align}
where $\mathbf{Q}$ is the covariance matrix of $\vec X$, i.e., $\mathbf{Q}= \mathbb E \left( \vec X \vec X^{\dag}\right)$. 

For transmission rate $R$, the throughput is $\mathcal R_{s} = \overline{\mathcal P}_{\text{out}}(R) R$, where $\mathcal P_{\text{out}}(R)$ is the outage probability of a fixed transmission rate $R$. 
%In the other words, the throughput decreases with outage probability.
%This means that the less outage probability, the higher throughput. 
It is proved in \cite{jorswieck} that the optimum transmit strategy minimizing the outage probability is to send uncorrelated circularly symmetric zero mean equal power complex Gaussian signals from a fraction of antennas. Thus, here, one can restrict the transmit covariance matrix $\mathbf{Q}$ to diagonal matrices whose diagonal entries are either zero or a constant subject to the total power constraint $P$.

In following, \cref{miso-single-theorem} proves that the optimum solution with respect to $R$, denoted by $R^o$, maximizing $\overline{\mathcal P}_{\text{out}}(R) R$ is less than $\ln \left( 1+P \right)$. In this range of the transmission rate, the optimum transmit strategy which minimizes the outage probability and consequently, maximizes the throughput is to use all available antennas. \Cref{miso-sl-rate-max} yields the maximum throughput of an $n_t \times 1$ MISO block Rayleigh fading channel.

%\Cref{miso-single-theorem} presents the maximum throughput of the point-to-point block fading channel with a $n_t$-antenna transmitter and a single-antenna receiver, and proves that the optimum transmit strategy maximizing \cref{miso_single1} is to transmit equal energy uncorrelated signals across all available antennas, i.e., $\mathbf{Q^o}=\frac{P}{n_t} \mathbf I_{n_t}$.

\begin{theorem} \label{miso-single-theorem}
In a single-layer $n_t \times 1$ MISO block Rayleigh fading channel, the optimum transmit covariance matrix which maximizes the throughput is $\mathbf{Q^o}=\frac{P}{n_t} \mathbf I_{n_t}$. 
The maximum throughput is given by
\begin{align} \label{miso-sl-rate-max}
\mathcal R_{s}^m = \max_{0 < s < 1} \frac{\Gamma(n_t,n_t s)}{(n_t-1)!} \ln\left( 1+P s \right).
\end{align}
\end{theorem}

\begin{proof}

As pointed out above, we can restrict our attention to assume that $l_t$ out of $n_t$ transmit antennas are active and perform equal power allocation. \Cref{miso_single1} is simplified to

\begin{align} \label{miso_single_Rayleigh3}
   \mathcal R_{s}^m &= \max_{R,l_t}  \Pr\left\{\ln \left( 1+\frac{P}{l_t} \sum_{\ell=1}^{l_t} |h_\ell|^2 \right) \geq R \right\}R  \nonumber \\
  &=\max_{s,l_t}  \Pr\left\{ \sum_{\ell=1}^{l_t} |h_\ell|^2  \geq l_t s \right\}R \nonumber \\
  &=  \max_{s,l_t} \overline F_{\mathrm a} (l_t s) \ln \left( 1+P s  \right),
\end{align}
where $a \Def \sum_{\ell=1}^{l_t} | h_\ell|^2 $ is gamma-distributed and thereby, $\overline F_{\mathrm a}(x)=\frac{\Gamma (l_t, x) }{\Gamma(l_t)}$. The first derivative of $\mathcal R_s (s)=\overline F_{\mathrm a} (l_t s) \ln \left( 1+P s  \right)$ with respect to $s$ is 
\begin{align}
{\mathcal R}_s^{\prime}(s) = \overline F_{\mathrm a} (l_t s) \frac{P}{1+P s} - l_t f_{\mathrm a}(l_t s) \ln \left( 1+P s  \right).
\end{align}
Let us define the following functions,
\begin{align}
&r(s) \Def \frac{\overline F_{\mathrm a}(l_t s)}{l_t f_{\mathrm a}(l_t s)}, \\
&g(s,P) \Def \ln \left( 1+P s \right)^{\frac{1+P s }{P}}. 
\end{align}
As such, we get
\begin{align} \label{diff_compair_trg}
\left\{
\begin{array}{lcr}
{\mathcal R}_s^{\prime}(s) > 0  & \text{iff} & r(s)>g(s,P), \\
{\mathcal R}_s^{\prime}(s) = 0  & \text{iff} & r(s)=g(s,P), \\
{\mathcal R}_s^{\prime}(s) < 0  & \text{iff} & r(s)<g(s,P). 
\end{array}
\right. 
\end{align}
Noting $\overline F_{\mathrm a}(x)=\frac{\Gamma (l_t, x) }{\Gamma(l_t)}$ and $f_{\mathrm a}(x)=\frac{x^{l_t-1}e^{-x}}{\Gamma(l_t)}$, we have
\begin{align} \label{lhs_diff_miso_single_Rayleigh1}
r(s) = \frac{\Gamma (l_t, l_t s)}{l_t (l_t s)^{l_t-1}e^{-l_t s}}=\frac{\Gamma (l_t, l_t s)}{l_t^{l_t} s^{l_t-1}e^{-l_t s}}.
\end{align}
For positive integer arguments of $m$, $\Gamma (m,x)=(m-1)!e^{-x}\sum_{\ell=0}^{m-1} \frac{x^\ell}{\ell!}$. Inserting the above equation into \cref{lhs_diff_miso_single_Rayleigh1} yields
\begin{align} \label{lhs_diff_miso_single_Rayleigh2}
r(s) &= \frac{(l_t-1)!e^{-l_t s}\sum_{\ell=0}^{l_t-1} \frac{(l_t s)^\ell}{\ell!}}{l_t (l_t s)^{l_t-1}e^{-l_t s}} \nonumber \\
&=\frac{1}{l_t}+ \frac{1}{l_t} \sum_{\ell=0}^{l_t-2} \frac{(l_t-1)\dots(\ell+1)}{(l_t s)^{l_t-\ell-1}} \nonumber \\
&=
\frac{1}{l_t}+\frac{1}{l_t} \sum_{\ell=0}^{l_t-2} \prod_{k=0}^{l_t-\ell-2} \frac{l_t-k-1}{l_t s}.
\end{align}
As $\frac{l_t-k-1}{l_t s} < 1$ for $s \geq 1$, replacing in \cref{lhs_diff_miso_single_Rayleigh2} gives
\begin{align} \label{lhs_diff_miso_single_Rayleigh3}
r(s) \leq \frac{1}{l_t}+\frac{1}{l_t} \sum_{\ell=0}^{l_t-2} \prod_{k=0}^{l_t-\ell-2} 1 =
\frac{1}{l_t}+\frac{l_t-1}{l_t}= 1,~~ \forall s \geq 1.
\end{align}
From \cref{lhs_diff_miso_single_Rayleigh2}, $\lim_{s\to 0} r(s) = +\infty$.
%\begin{align} \label{lhs_diff_miso_single_Rayleigh4}
%\lim_{s\to 0} r(s) = +\infty.
%\end{align}

On the other hand, the first derivative of $g\left(s\right)$ with respect to $P$ is 
\begin{align}
\frac{ \partial g(s,P)}{\partial P} &= \frac{s P-\ln\left( 1+ s P  \right)}{P^2} \nonumber \\
&= \frac{1}{P^2}\ln \frac{e^{s P}}{1+s P} \nonumber \\
&=\frac{1}{P^2}\ln \left( 1+\frac{1}{1+s P} \sum_{k=2}^{\infty} \frac{\left(s P \right)^k}{k!} \right) > 0.
\end{align}
Therefore, $g(s,P)$ is a strictly increasing function with respect to $P$. As a result,
\begin{align}  \label{lhs_diff_miso_single_Rayleigh5}
 g(s,P) > \lim_{P \to 0} \ln \left( 1+ P s\right)^{\frac{1+P s}{P}} = s.
\end{align}
Comparing \cref{lhs_diff_miso_single_Rayleigh3}, \cref{lhs_diff_miso_single_Rayleigh5}, $\lim_{s\to 0} r(s) = +\infty$, and $g(0,P)=0$, we get
\begin{align} \label{lhs_diff_miso_single_Rayleigh6}
\left\{
\begin{array}{lc}
r(s)>g(s,P) & s=0, \\
r(s)<g(s,P) & s \geq 1.
\end{array}
\right. 
\end{align}
Inserting \cref{lhs_diff_miso_single_Rayleigh6} into \cref{diff_compair_trg} yields
\begin{align} \label{diff_miso_single_Rayleigh1}
\left\{
\begin{array}{lc}
{\mathcal R}_s^{\prime}(s)>0 & s=0, \\
{\mathcal R}_s^{\prime}(s)<0 & s \geq 1.
\end{array}
\right. 
\end{align}

Since $\mathcal R_s(s)$ is a continuous function, according to \cref{diff_miso_single_Rayleigh1}, for all positive integer values of $l_t$ and positive values of $P$, one can conclude that $\mathcal R_s(s)$ takes its maximum at $0 < s^o < 1$.

Jorswieck and Boche \cite{jorswieck} proved that when $P > e^R-1$, or equivalently $s < 1$, the optimum transmission strategy to minimize the outage probability is to use all available antennas with equal power allocation. Since $\forall l_t$, $0 < s^o <1$, the optimum strategy maximizing the throughput is to use all available antennas and perform equal power allocation. 
The maximum throughput is given by \cref{miso-sl-rate-max}.
%\begin{align}
%\mathcal R_{s}^m = \max_{0 \leq s \leq 1} \Gamma(n_t,n_t s) \ln\left( 1+P s \right)
%\end{align}

\end{proof}

\begin{remark}
In point-to-point single-input single-output (SISO) channels, by substituting $n_t=1$ in \cref{miso-sl-rate-max}, the optimum solution with respect to $s$ is $s^o=\frac{1}{\mathcal W_0\left(P\right)}-\frac{1}{P}$, where $\mathcal W_0 \left( \cdot \right)$ is the zero branch of the Lambert W-function. Therefore, 
\begin{align}
\mathcal R_{s}^m = e^{\frac{1}{P}-\frac{1}{\mathcal W_0\left(P\right)}} \ln \left(\frac{P}{\mathcal W_0\left(P\right)}  \right).
\end{align}
From \cref{ergodic-capacity-miso-pro}, the ergodic capacity in this channel is 
\begin{align}
C_{\emph{erg}} = e^{\frac{1}{P}} \emph{E}_1\left( \frac{1}{P} \right).
\end{align}
\end{remark}
%\begin{description}
%\item{Note:}
%\textbf{Note:} In a MISO $n_t \times 1$ block fading channel with $\vec h \sim \left[\mathcal{CN}(0,1)\right]_{1 \times n_t}$ (Rayleigh-fading), $s^o$ is a monotonically decreasing function of the transmission power, $P$. This comes from the fact that $g(s,P)$ is a monotonically increasing function of $P$, while $r(s)$ is independent of $P$. Therefore, increasing $P$ makes $r(s)<g(s,P)$ and as a result, $\mathring T(s)<0$, which results in decreasing $s^o$. 
%\end{description}

\begin{remark} \label{remark-misi-throughput-asymptote}
Note that $g\left(s,P\right)$ is a strictly increasing function with respect to $s$ and $P$, and $r\left(s\right)$ is a strictly decreasing function with respect to $s$ and increases with the number of transmit antennas. Therefore, the solution to $r\left(s\right)=g\left(s,P\right)$, i.e., $s^o$, 
\begin{itemize}
\item decreases with $P$. In asymptotically high SNR regime, $s^o \to 0$.
\item increases with $n_t$. In asymptotically large number of transmit antennas, $s^o \to 1$.
\end{itemize}
\end{remark}

As a byproduct result of \cref{miso-single-theorem,remark-misi-throughput-asymptote}, we have the following.

\begin{corollary}
In the asymptotically large number of transmit antennas MISO channel, the maximum throughput is given by
\begin{align}
\mathcal R_s^m = \lim_{s\to 1} \frac{\Gamma\left( n_t,n_t s \right)}{\left( n_t-1 \right)!} \ln\left( 1+P s \right) \stackrel{n_t\to \infty}{\longrightarrow} \ln\left( 1+P \right).
%e^{-n_t} \sum_{\ell=0}^{n_t-1} \frac{n_t^\ell}{\ell!} \ln \left( 1+P \right).
\end{align}
\end{corollary}

%\Cref{miso-single-fig} presents the maximum throughput of the MISO channel for $n_t=1$ -- $6$ and $P=$-$10$dB -- $50$dB. It shows that for all ranges of SNR, the maximum troughput increases by the use of more transmit antennas.
%
%\begin{figure}
%\centering
%%\epsfig{file=fig3_.eps,width=0.25\linewidth,clip=}
%\includegraphics[scale=0.7]{single-layer2.eps}
%%\includegraphics[height=4cm, width=.7\linewidth]{single.pdf}
%\caption{Maximum throughput of the MISO channel.}
%\label{miso-single-fig}
%\end{figure}

\begin{remark} \label{remark-correlation-misi-single}
In a correlated MISO channel wherein the transmitter does neither know the CSI nor the channel correlation, the outage probability is a Schur-convex (resp. Schur-concave) function of the channel covariance matrix for $P>e^R-1$ (resp. $P<\frac{e^R-1}{2}$) \cite{jorswieck}.
According to \cref{miso-single-theorem}, in the maximum throughput of the MISO channel, i.e., $\overline{\mathcal P}_{out}(R^o)R^o$, we have $e^{R^o}-1<P$. Hence, in this range of the transmission rate, $\mathcal R_s$ is a Schur-concave function of the channel covariance matrix, i.e., channel correlation decreases the throughput. In terms of the impact of correlation in the MISO channel with no CSI at the transmitter, the behavior of the maximum throughput is similar to the behavior of the ergodic capacity which is also a Schur-concave function of the channel covariance matrix \cite{boche2004ergodic}.
\end{remark}

%%%%%%%%%%%%%%%%%%%%%%%%%%%%%%%%%%%%%%%%%%%%%%%%%%%%%%%%%%%%%%%%%%%%%%%%%%%%%%%%%%%%%%%%%%%%%%%%%%%%%%%%%%%%%%%%%%%
%%%%%%%%%%%%%%%%%%%%%%%%%%%%%%%%%%%%%%%%%%%%%%%%%%%%%%%%%%%%%%%%%%%%%%%%%%%%%%%%%%%%%%%%%%%%%%%%%%%%%%%%%%%%%%%%%%%
%%%%%%%%%%%%%%%%%%%%%%%%%%%%%%%%%%%%%%%%%%%%%%%%%%%%%%%%%%%%%%%%%%%%%%%%%%%%%%%%%%%%%%%%%%%%%%%%%%%%%%%%%%%%%%%%%%%
%%%%%%%%%%%%%%%%%%%%%%%%%%%%%%%%%%%%%%%%%%%%%%%%%%%%%%%%%%%%%%%%%%%%%%%%%%%%%%%%%%%%%%%%%%%%%%%%%%%%%%%%%%%%%%%%%%%
%%%%%%%%%%%%%%%%%%%%%%%%%%%%%%%%%%%%%%%%%%%%%%%%%%%%%%%%%%%%%%%%%%%%%%%%%%%%%%%%%%%%%%%%%%%%%%%%%%%%%%%%%%%%%%%%%%%

\section{Maximum Expeted-Rate in MISO Channels} \label{miso-multi-section}

A block fading channel can be modeled by an equivalent broadcast channel whose receiver channels represent any fading coefficient realization. The expected-rate of a fading channel is equal to a weighted sum-rate of its equivalent broadcast channel in which the weights distribution is the complementary CDF (tail distribution) of the channel gain \cite{shamai1997broadcast}. 
In broadcast channels, any maximum weighted sum-rate with positive value weights is on the capacity region \cite{reza}. 
Since superposition (multi-layer) coding achieves the capacity region of degraded broadcast channels \cite{thomas2006}, it is the optimum coding strategy to maximize the average achievable rate in any block fading channel whose equivalent broadcast channel is degraded \cite{shamai}. An example for such channels is the SISO channel. Although multi-layer coding is not the optimum coding strategy in MISO channels, it increases the average achievable rate of the channel. Numerical results for the continuous-layer expected-rate of MISO and SIMO block Rayleigh fading channels were presented in \cite{steiner2007multi}. Here, the optimum transmit covariance matrix at each code layer is obtained, and consequently, the maximum expected-rate of the MISO channel is analytically formulated. Note that the maximum expected-rate of the SIMO channel can be calculated using the same formula by replacing $P$ with $n_t P$ in \cref{miso-integral-limits-cl}.

In order to enhance the lucidity of this section, we divide it into two subsections. 
%first we analyze finite-layer coding for MISO channels. Afterwards, continuous-layer expected-rate is derived.
\Cref{miso-multi_finite-subsection} presents the maximum expected-rate of the MISO channel when a finite-layer code is transmitted. 
The more code layers, the higher expected-rate. Hence, a continuous-layer (infinite-layer) code yields the highest expected-rate of the channel. 
%In this approach, a continuum of code layers is transmitted, and the receiver decodes from the lowest layer up to the layer that the channel condition allows. 
The maximum continuous-layer expected-rate of the MISO channel is derived in \cref{miso-multi_infinite-subsection} in closed form. 

\subsection{Finite-Layer Code} \label{miso-multi_finite-subsection}

In finite-layer coding approach, the transmitter sends a $K$-layer code $\vec X=\sum_{i=1}^K \vec X_i$. Let $P_i$ be the signal power in the $i$'th layer with rate $R_{i}=\ln \left( 1+\frac{P_i s_i}{1+ I_i s_i} \right)$, where $I_i=\sum_{j=i+1}^K P_j$ is the power of the upper layers while decoding the $i$'th layer. 
The maximum expected-rate in \cref{finite-layer-expected-rate-definition-formula} is simplified to
\begin{align} \label{miso_multii1}
\mathcal R_{f}^m {=} \!\! \max_{\begin{subarray}{c}
R_i,P_i,\mathbf{Q}_i \\
\text{tr}(\mathbf{Q}_i)\leq P_i \\
\sum_{i=1}^K P_i = P
\end{subarray}} 
\sum_{i=1}^K
\Pr\! \left\{\! \ln \left(\!1\!+\!\frac{\vec h \mathbf{Q}_i \vec h^{\dag}}{\vec h \sum_{j=i+1}^K \mathbf{Q}_j \vec h^{\dag}} \!\right) \!\geq R_i  \!\right\}\!R_i.
\end{align}

\Cref{miso-multi-theorem} presents the optimum covariance matrix in each layer which maximizes the expected-rate in the MISO channel.

\begin{theorem} \label{miso-multi-theorem}
In a finite-layer $n_t \times 1$ MISO block Rayleigh fading channel, the optimum transmit covariance matrix in each layer which maximizes the expected-rate is $\mathbf Q_i^o=\frac{P_i}{n_t} \mathbf I_{n_t}$, where $P_i$ is the power allocated to the $i$'th layer. The maximum $K$-layer expected-rate is given by
\begin{align} \label{miso_avg_max_ml}
\mathcal R_{f}^m {=}\! \max_{\begin{subarray}{c}
0 < s_i < 1, P_i \\
%0 \leq P_l \leq P \\
\sum_{i=1}^K P_i = P
\end{subarray} }
\sum_{i=1}^K \frac{\Gamma \left( n_t,n_t s_i\right)}{(n_t-1)!} \ln \!\left(\! 1\!+\!\frac{P_i s_i}{1+\sum_{j=i+1}^K P_j s_i}\! \right).
\end{align}
\end{theorem}

\begin{proof}
%See Apendix A.
%Thanks to the results of \cite{visotsky}, 
Since the outage probability does not depend on the directions of the transmit covariance matrix $\mathbf Q$ \cite{visotsky}, the problem is diagonalized. 
%Let $P_i$ be the signal power of the $i$'th layer with rate $R_{i}=\ln \left( 1+\frac{P_i s_i}{1+ I_i s_i} \right)$, where $K$ is the number of layers, and $I_i=\sum_{j=i+1}^K P_j$ is the interference of upper layers while decoding the $i$'th layer. 
Therefore, the expected-rate received at the destination is simplified to
\begin{align} \label{miso_multi1}
\mathcal R_{f} = \sum_{i=1}^K \Pr \left\{ \ln \left(1+\frac{P_i \sum_{\ell=1}^{n_t} \delta_\ell |h_\ell|^2}{1+I_i \sum_{\ell=1}^{n_t} \eta_\ell |h_\ell|^2} \right) \geq R_i  \right\}R_i,
\end{align}
where $\delta_\ell$ and $\eta_\ell$ are the power fraction and upper-layer interference portion at the $\ell$'th antenna, respectively, subject to $\sum_{\ell=1}^{n_t} \delta_\ell = \sum_{\ell=1}^{n_t} \eta_\ell = 1$. \Cref{miso_multi1} can be rewritten as 
\begin{align} \label{miso_multi_proof1}
\mathcal R_{f} = 
%\Pr \left\{ \frac{\sum_{i=1}^{n_t} P_i |h_i|^2}{1+\sum_{i=1}^{n_t} I_i |h_i|^2} \geq \frac{P_l s_{l}}{1+I_l s_l} \right\}\ln \left( 1+\frac{P_l s_{l}}{1+I_l s_l} \right)= \nonumber \\
\sum_{i=1}^K \Pr \left\{\sum_{\ell=1}^{n_t} \left( \delta_\ell+s_i I_i \delta_\ell -s_i I_i \eta_\ell \right) |h_\ell|^2 \geq s_i \right\}R_i.
\end{align}
As $\sum_{\ell=1}^{n_t} \left( \delta_\ell+s_i I_i \delta_\ell -s_i I_i \eta_\ell \right) = 1$, to minimize $\Pr \left\{\sum_{\ell=1}^{n_t} \left( \delta_\ell+s_i I_i \delta_\ell -s_i I_i \eta_\ell \right) |h_\ell|^2 < s_i \right\},~\forall i$, the optimum value of $\delta_\ell+s_i I_i \delta_\ell -s_i I_i \eta_\ell$ must be either zero or a constant independent of $\ell$ for any positive value of $s_i$.
%, with respect to the proof of Theorem 3 in Section 3.1 of \cite{jorswieck}. 
Hence, up to now, the optimum solution to \cref{miso_multi_proof1} is to choose either $\delta_\ell=\eta_\ell=\frac{1}{l_{t_i}}$ or $\delta_\ell=\eta_\ell=0$, that is to use $l_{t_i}$ out of $n_t$ antennas with power $\frac{P_i}{l_{t_i}}$ in each layer. Therefore, \cref{miso_multi_proof1} is simplified to
\begin{align} \label{miso_multi1_2}
\mathcal R_{f} = \sum_{i=1}^K \Pr \left\{\sum_{\ell=1}^{l_{t_i}}  |h_\ell|^2 \geq l_{t_i} s_i  \right\}R_i=
\sum_{i=1}^K \overline F_{\mathrm a_i} \left( l_{t_i} s_i \right) R_i,
\end{align}
where $a_i=\sum_{\ell=1}^{l_{t_i}}  |h_\ell|^2$.
In the remainder of the proof, we shall show that the optimum solution with respect to $l_{t_i}$ is $l_{t_i}^o=n_t,~\forall i$.
Analogous to the throughput case in \cref{miso-single-theorem}, let us define 
\begin{align}
&\mathcal R_s(s_i) \Def \overline F_{\mathrm a_i} \left( l_{t_i} s_i \right) \ln \left( 1+\frac{P_i s_{i}}{1+I_i s_i} \right), \\
&r(s_i) \Def \frac{\overline F_{\mathrm a_i}(l_{t_i} s_i)}{l_{t_i} f_{\mathrm a_i}(l_{t_i} s_i)}, \\
&g(s_i,P_i,I_i) \Def \frac{\left(1+I_i s_i \right)\left(1+\left( I_i+P_i\right) s_i \right)}{P_i} \ln \left( 1+\frac{P_i s_{i}}{1+I_i s_i} \right). 
\end{align}
Note that $g(0,P_i,I_i)=0$, $\lim_{s_i \to 0} r(s_i) = +\infty$, and \cref{diff_compair_trg,lhs_diff_miso_single_Rayleigh3} still hold by redefining $\mathcal R_s(s_i)$, $r(s_i)$, and $g(s_i,P_i,I_i)$ as above, and with $s$ replaced by $s_i$. 

%Noting $g(s_i,P_i,I_i)$ is a monotonically increasing function of $s_i$, for $s_i \geq 1$, we get
Defining $\hat P_i \Def \frac{P_i}{1+I_i s_i}$, from \cref{lhs_diff_miso_single_Rayleigh5} and noting $I_i s_i \geq 0$, we have
\begin{align}
g(s_i,P_i,I_i) &= \left( 1+I_i s_i \right) \frac{\left(1+\frac{P_i s_i
}{1+I_i s_i} \right)}{\frac{P_i}{1+I_i s_i}} \ln \left( 1+\frac{P_i s_i
}{1+I_i s_i} \right) \nonumber \\
& {\geq} 
%\frac{\left(1+\frac{P_i s_i
%}{1+I_i s_i} \right)}{\frac{P_i}{1+I_i s_i}} \ln \left( 1+\frac{P_i s_i
%}{1+I_i s_i} \right) \nonumber \\
%&=
\ln \left( 1+\hat P_i s_i \right)^{\frac{\left(1+\hat P_i s_i \right)}{\hat P_i}}
{>} s_i,~~~ \forall s_i \geq 1.
\end{align}
%\begin{align}
%&g(s_i,P_i,I_i) = \left( 1+I_i s_i\right) \frac{\left(1+\frac{P_i s_i
%}{1+I_i s_i} \right)}{\frac{P_i}{1+I_i s_i}} \ln \left( 1+\frac{P_i s_i
%}{1+I_i s_i} \right) \nonumber \\
%& \stackrel{(a)}{\geq} \frac{\left(1+\frac{P_i s_i
%}{1+I_i s_i} \right)}{\frac{P_i}{1+I_i s_i}} \ln \left( 1+\frac{P_i s_i
%}{1+I_i s_i} \right) \stackrel{(b)}{\geq} s_i, \forall s_i \geq 1
%\end{align}
%where (a) and (b) are the direct results of $I_i s_i\geq 0$ and \cref{lhs_diff_miso_single_Rayleigh5}, respectively.

Therefore, \cref{lhs_diff_miso_single_Rayleigh6,diff_miso_single_Rayleigh1} still hold with the above functions, and lead to $0 < s_i^o < 1$. 
This directly corresponds to the proof of \cref{miso-single-theorem} and shows that the optimum power allocation strategy is to use all available antennas with equal power allocation in each layer, i.e., $\mathbf Q_i^o=\frac{P_i}{n_t} \mathbf I_{n_t}$, and the maximum expected-rate is given by \cref{miso_avg_max_ml}. 

\end{proof}

%According to Theorem 2, the maximum expected-rate of continuous-layer coding at the transmitter is derived from
%\begin{align} \label{r_av_shamai} 
%\mathcal R_{c}^m = \int_{s_0}^{s_1}  \overline F_{a}(s)\left(\frac{2}{s}+\frac{\mathring{f}_{a}(s)}{f_{a}(s)} \right) \ud s,
%\end{align}
%where boundaries $s_0$ and $s_1$ are the solutions of $\frac{\overline F_{a_{eq}}(s_0)}{s_0 f_{a_{eq}}(s_0)} =1+\frac{P}{n_t} s_0$ and $\frac{\overline F_{a_{eq}}(s_1)}{s_1 f_{a_{eq}}(s_1)} =1$, respectively, where $a_{eq} = \sum_{i=1}^{n_t} |h_i|^2$.
%\begin{align}  \label{marginal_condition_shamai2}
%\left\{ \begin{array}{lc} 
%\frac{\overline F_{a_{eq}}(s_0)}{s_0 f_{a_{eq}}(s_0)} =1+\frac{P}{n_t} s_0   \\
%\frac{\overline F_{a_{eq}}(s_1)}{s_1 f_{a_{eq}}(s_1)} =1
%\end{array}
%\right.,
%\end{align}
%and $\overline F_{a_{eq}}(x)=\frac{\Gamma (n_t, n_t x) }{\Gamma(n_t)}$.

%The maximum expected-rate of the MISO channel using a two-layer code at the transmitter is calculated numerically and depicted in \cref{miso-two-fig} for various number of transmit antennas. 
%%Similar to the maximum throughput (\cref{miso-single-fig}), one can see that the more transmit antennas, the more maximum 
%
%\begin{figure}
%\centering
%%\epsfig{file=fig3_.eps,width=0.25\linewidth,clip=}
%\includegraphics[scale=0.7]{two_layer.eps}
%%\includegraphics[height=4cm, width=.7\linewidth]{single.pdf}
%\caption{Maximum expected-rate of the MISO channel using two-layer coding.}
%\label{miso-two-fig}
%\end{figure}

%%%%%%%%%%%%%%%%%%%%%%%%%%%%%%

\subsection{Continuous-Layer Code} \label{miso-multi_infinite-subsection}

In the continuous-layer coding, a.k.a.\ broadcast approach, a continuum of code layers is transmitted.
Similar to finite-layer coding in \cref{miso-multi_finite-subsection}, the receiver decodes the signal from the lowest layer up to the layer that the channel condition allows. 
%The maximum expected-rate can be written as \cite{shamai}
%\begin{align} \label{general_shamai_rate_formula}
%\mathcal R_{c}^m = \max_{I(x)} \int_{0}^{\infty} \overline F_{a_{eq}}(x) \frac{-x\mathring I(x)}{1+x I(x)} \ud x.
%\end{align}
%Maximizing \cref{general_shamai_rate_formula} over $I(x)$ using variation methods \cite{calculus} results in
%\begin{align} \label{general_shamai_rate_formula_answer}
%\mathcal R_{c}^m = \int_{x_0}^{x_1} \overline F_{a_{eq}}(x) \left( \frac{2}{x} + \frac{\mathring f_{a_{eq}}(x)}{f_{a_{eq}}(x)} \right) \ud x.
%\end{align}
%The integration limits are the solutions to 
%\begin{align}
%\begin{cases}
%\frac{\overline F_{a_{eq}}(x_0)}{x_0 f_{a_{eq}}(x_0)} = 1+\frac{P}{n_t} x_0, \\
%\frac{\overline F_{a_{eq}}(x_1)}{x_1 f_{a_{eq}}(x_1)} \end{align}

\Cref{miso-infinite-proposition} yields a closed form expression for the maximum continuous-layer expected-rate in the MISO channel by optimizing the power distribution over the layers.

\begin{proposition} \label{miso-infinite-proposition}
In the MISO block Rayleigh fading channel, the maximum continuous-layer expected-rate obtained by optimizing the power distribution over the layers is given by 
\begin{align} \label{indefinite-rate-miso-cl}
\mathcal R_{c}^m = \mathcal R(s_1) - \mathcal R(s_0),
\end{align}
where,
%\begin{align}
%g(x) = e^{-x} \sum_{\ell = 0}^{n_t-1} \frac{x^{\ell}}{\ell !} - e^{-x} \sum_{\ell = 1}^{n_t-1} \frac{n_t+ 1-\ell}{\ell} \sum_{i = 0}^{\ell-1} \frac{x^i}{i!} \nonumber \\
%-(n_t+1) \emph{E}_1 (x) 
%\end{align}
\begin{align} \label{indefinite-integral-miso-cl}
\mathcal R(s) = e^{-s}  \sum_{\ell = 1}^{n_t-1} \frac{1}{\ell !} \left( s^{\ell} - (n_t+ 1-\ell)(\ell -1)! \sum_{k = 0}^{\ell-1} \frac{s^k}{k!}  \right) \nonumber \\
+e^{-s}-(n_t+1) \emph{E}_1 (s).
\end{align}
$s_0$ and $s_1$ are the solutions to 
\begin{align} \label{miso-integral-limits-cl}
\begin{cases}
\sum_{\ell=0}^{n_t-1} \frac{(n_t-1)!}{\ell ! s_0^{n_t-\ell}} = 1+\frac{P}{n_t} s_0, \\
\sum_{\ell=0}^{n_t-1} \frac{(n_t-1)!}{\ell ! s_1^{n_t-\ell}} = 1,
\end{cases} 
\end{align}
respectively.
\end{proposition}

\begin{proof}
Based on \cref{miso-multi-theorem}, transmitting each of the code layers on all available antennas and performing equal power allocation is optimum. 
As showed in \cite{shamai}, the maximum continuous-layer expected-rate of fading channels with general distribution is given by 
\begin{align} \label{infinite-layer-general-me}
\mathcal R_{c}^m = \max_{I(s)} \int_{0}^{\infty} \overline F_{\mathrm a}(s) \frac{-s I^{\prime}(s)}{1+sI(s)} \ud s.
\end{align}
Noting $\overline F_{\mathrm a}(s)=\frac{\Gamma\left( n_t, s \right)}{\Gamma\left( n_t \right)}=e^{-s} \sum_{\ell=0}^{n_t-1} \frac{s^{\ell}}{\ell!}$, we have
\begin{align} \label{infinite-layer-general-shamai}
\mathcal R_{c}^m = \max_{I(s)} \int_{0}^{\infty} \frac{-se^{- s} I^{\prime}(s)}{1+sI(s)} \sum_{\ell=0}^{n_t-1} \frac{s^\ell}{\ell!} \ud s.
\end{align}
%where $\overline F_{a_{eq}}(x) = \frac{\Gamma \left( n_t,n_t s_l\right)}{\Gamma \left( n_t\right)}$. 
The optimization solution to \cref{infinite-layer-general-shamai} with respect to $I(s)$ under the total power constraint $\frac{P}{n_t}$ at each antenna is found using variation methods \cite{calculus}. By solving the corresponding E$\ddot{\text{u}}$ler equation \cite{calculus}, we come up with the final solution as follows,
%\begin{align} \label{infinite-layer-general-me}
%\mathcal R_{c}^m = \int_{x_0}^{x_1}  \overline F_{a_{eq}}(x)\left(\frac{2}{x}+\frac{\mathring{f}_{a_{eq}}(x)}{f_{a_{eq}}(x)} \right) \ud x,
%\end{align}
%where boundaries $x_0$ and $x_1$ are the solutions of $\frac{\overline F_{a_{eq}}(x_0)}{x_0 f_{a_{eq}}(x_0)} =1+\frac{P}{n_t} x_0$ and $\frac{\overline F_{a_{eq}}(x_1)}{x_1 f_{a_{eq}}(x_1)} =1$, respectively.
%, where $a_{eq} = \sum_{i=1}^{n_t} |h_i|^2$.
%After algebraic simplifications, we have
%Substituting in \cref{general_shamai_rate_formula_answer}, we get
\begin{align} \label{infinite-layer-miso-integral}
\mathcal R_{c}^m = \int_{s_0}^{s_1}e^{-s} \left( \frac{n_t+1}{s} -1  \right) \sum_{\ell =0}^{n_t-1} \frac{s^\ell}{\ell !} \ud s,
\end{align}
where boundaries $s_0$ and $s_1$ are the solutions to $\sum_{\ell=0}^{n_t-1} \frac{(n_t-1)!}{\ell ! s_0^{n_t-\ell}} = 1+\frac{P}{n_t} s_0$ and $\sum_{\ell=0}^{n_t-1} \frac{(n_t-1)!}{\ell ! s_1^{n_t-\ell}} = 1$, respectively.
The indefinite integral (antiderivative) of \cref{infinite-layer-miso-integral} is given by \cref{indefinite-integral-miso-cl} (the derivation steps are deferred to appendix \ref{proof-expected-rate-integral}). Applying the integration limits completes the proof. 

\end{proof}

\begin{remark}
By substituting $n_t=1$ in \cref{miso-infinite-proposition}, the maximum continuous-layer expected-rate of the SISO channel is
\begin{align} \label{r_av_shamai_Rayleigh} 
\mathcal R_{c}^m = 2\emph{E}_1\left(\frac{2}{1+\sqrt{1+4P}}\right)-2\emph{E}_1(1)- e^{\frac{-2}{1+\sqrt{1+4P}}}+e^{-1}.
\end{align}
As pointed out earlier, one can model a point-to-point block Rayleigh fading channel with an equivalent broadcast channel. 
According to the degradedness of the equivalent SISO broadcast channel, and the optimality of superposition (multi-layer) coding for such channels \cite{thomas2006}, the maximum continuous-layer expected-rate of the SISO channel, i.e., \cref{r_av_shamai_Rayleigh}, represents its maximum average achievable rate \cite{shamai}.
\end{remark}

\begin{remark}
Since the equivalent broadcast channel of the MISO channel is not degraded, its maximum continuous-layer expected-rate is not the maximum average achievable rate of the channel. For example, in asymptotically low SNR regime, the multiple-access scheme provides a higher average achievable rate in the MISO channel. In the multiple-access scheme, the antennas send independent messages, and the receiver decodes as much as it can.
\end{remark}

\begin{remark}
Similar to \cref{remark-correlation-misi-single}, one can conclude that for $0 < s_i^o < 1,~\forall i$, the maximum expected-rate of the MISO channel with uninformed transmitter is a Schur-concave function of the channel covariance matrix, that is channel correlation reduces the maximum expected-rate.
\end{remark}

\section{Maximum Throughput in MIMO Channels} \label{mimo-section}

The throughput maximization problem in the MIMO channel is less tractable than that corresponding to the MISO channel.

Since in the Gaussian MIMO channel, in the sense of the outage probability, the optimum eigenvectors of the transmit covariance matrix always correspond to the eigenvectors of the channel correlation matrix \cite{visotsky}, one can restrict the transmit covariance matrix to be diagonal in the problem of interest. 

Recall from \cref{problem-setup}, in an $n_t \times n_r$ MIMO channel, the PDF of the instantaneous mutual information in \cref{instantaneous-mutual-information-mimo-general} does not lend itself to a closed form expression.
%\begin{align}
%\mathcal I = \ln \det \left( \mathbf I_{n_t} + \mathbf Q \mathbf H^\dag \mathbf H \right)= \ln \det \left( \mathbf I_{n_r} + \mathbf H \mathbf Q  \mathbf H^\dag \right),
%\end{align}
%where $\mathbf Q$ and $\mathbf H$ are the transmit covariance matrix and channel coefficient matrix, respectively.
In order to analyze the throughput, it is necessary to characterize this PDF. There are some approximations for the PDF of the instantaneous mutual information in literature, e.g., approximations on the distribution of the eigenvalues of $\mathbf H \mathbf H^\dag$ in MIMO channels with asymptotically large number of antennas at both the transmitter and receiver sides \cite{silverstein1995empirical,chuah2002capacity}.
%$\mathbf H \mathbf H^\dag$
%and asymptotic distribution approximation for the eigenvalues of $\mathbf H \mathbf H^\dag$ in correlated MIMO channels \cite{chuah2002capacity}. 
%For the large number of $n_t$ and $n_r$ with a fixed $\frac{n_t}{n_r}$, the distribution of the eigenvalues of $\mathnf H \mathbf H^\dag$ 

In a MIMO channel with $\mathbf Q = \frac{P}{n_t} \mathbf I_{n_t}$, the PDF of the instantaneous mutual information can be well approximated by the Gaussian distribution with the same mean and variance \cite{wang,hochwald}, i.e.,
\begin{align}
\mathcal I  \sim \mathcal N \left( \mu(n_t,n_r) , \sigma^2(n_t,n_r)   \right),
\end{align}
where
\begin{align}
\begin{cases}
\mu(n_t,n_r) = \mathbb E \left( \mathcal I \right), \\
\sigma^2(n_t,n_r) = \text{Var}\left(\mathcal I\right).
\end{cases}
\end{align}
Note that $\mu(n_t,n_r)$ equals the ergodic capacity of an $n_t \times n_r$ MIMO channel, which is a strictly increasing function with respect to $n_t$ and $n_r$ \cite{telatar}.
%The maximum throughput thus reduces to
This Gaussian distribution approximation allows the throughput maximization to be expressed as
\begin{align} \label{mimo-max-throu-gaussian-first}
\mathcal R_{s}^m &= \max_R \Pr \left\{ \mathcal I \geq R  \right\} R \nonumber \\
&= \max_R \mathcal Q\left( \frac{R-\mu(n_t,n_r)}{\sigma(n_t,n_r)} \right)R.
%\nonumber \\
%&= \max_z \mathcal Q(z) \left( \sigma(n_t,n_r) z + \mu(n_t,n_r) \right)\label{throughput_gaussian_approx} \\
%&= \mathcal Q(z^o) \left( \sigma(n_t,n_r) z^o + \mu(n_t,n_r) \right). \label{throughput_gaussian_approx_a}
\end{align}
With $z = \frac{R-\mu(n_t,n_r)}{\sigma(n_t,n_r)}$, \cref{mimo-max-throu-gaussian-first} leads to
\begin{align} 
\mathcal R_{s}^m &= \max_z \mathcal Q(z) \left( \sigma(n_t,n_r) z + \mu(n_t,n_r) \right) \label{throughput_gaussian_approx} \\
&= \mathcal Q(z^o) \left( \sigma(n_t,n_r) z^o + \mu(n_t,n_r) \right), \label{throughput_gaussian_approx_a}
\end{align}
%The maximum is the solution to the derivative of \cref{throughput_gaussian_approx} with respect to $z$ equal to zero.
where $z^o$ is the solution to
\begin{align} \label{large_M_zo_eq} %\label{throughput_gaussian_approx_maxima}
-\frac{1}{\sqrt{2\pi}} e^{-\frac{z^{o^2}}{2}}\! \left( \sigma(n_t,\!n_r) z^o\! +\! \mu(n_t,\!n_r) \right)\! +\! \sigma(n_t,\!n_r) \mathcal Q(z^o) {=} 0.
\end{align}

%Let us define $p = \min \left\{ n_t , n_r \right\}$, $n = \max \left\{ n_t , n_r \right\}$, and
%\begin{align}
%\mathbf W = 
%\begin{cases}
%\mathbf H^{\dag} \mathbf H & n_t \leq n_r \\
%\mathbf H \mathbf H^{\dag} & n_t > n_r.
%\end{cases}
%\end{align}
%Matrix $\mathbf W$ has a central complex $p$-variate Wishart distribution with scale matrix $\mathbf I_{n_t}$ and $n$ degrees of freedom \cite{anderson2003}. Performing Bartlett decomposition, we get $\mathbf W=\mathbf A \mathbf A^\dag$, where $\mathbf A$ is a square lower triangular matrix (left triangular matrix) in the form of
%\begin{align}
%\mathbf A =
%\begin{bmatrix}
%a_{1,1} & 0 & 0 & \cdots & 0 \\
%a_{2,1} & a_{2,2} & 0 & \cdots & 0 \\
%a_{3,1} & a_{3,2} & a_{3,3} & \cdots & 0 \\
%\vdots & \vdots & \vdots & \ddots & \vdots \\
%a_{p,1} & a_{p,2} & a_{p,3} & \cdots & a_{p,p}
%\end{bmatrix},
%\end{align}
%where $a_{\ell,\ell}^2 \sim \Gamma \left(n-\ell+1,1 \right)$ and $a_{i,j} \sim \mathcal{CN}(0,1)$.
%Hence, $\det \mathbf W=\det \mathbf A \times \det \mathbf A^\dag = \prod_{\ell=1}^{p} a_{\ell,\ell}^2$.
%, where $\sqrt{c_\ell}$ is the $\ell$'th diagonal entry of $\mathbf A$ with distribution $c_\ell \sim \Gamma \left(n_{max}-\ell+1,1 \right)$. 

Since the existing approximations for the PDF of the instantaneous mutual information in the MIMO channel are not tractable enough to analyze the maximum throughput in general case, four asymptotic cases are investigated. In all four cases, it is shown that the optimum transmit strategy is to 
use 
%transmit equal power uncorrelated signals from 
all available antennas.
%covariance matrix is $\mathbf Q^o = \frac{P}{n_t} \mathbf I_{n_t}$. 
%It is worth to mention that 
It seems reasonable to conjecture that the above statement holds with the general MIMO channel.
To test the claim, 
\cref{mimo-single-general-fig} shows the maximum throughput in a MIMO channel with 10 receive antennas.
% and verifies the claim. 
Note that the number of transmit antennas varies from $1$ to $20$ and the total power $P$ sweeps the range of -$10$ dB to $50$ dB.

\begin{figure}
\centering
\includegraphics[scale=0.7]{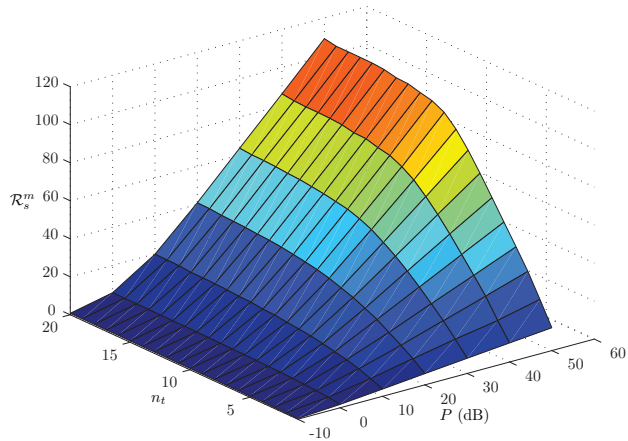}
\caption{The maximum throughput (in \emph{nats}) in a MIMO channel with 10 receive antennas ($n_r=10$).}
\label{mimo-single-general-fig}
\end{figure}

%%%%%%%%%%%%%%%%%%%%%%%%%%%%%%

\subsection{Asymptotically Low SNR Regime}

For small SNR values, the eigenvalues of $\mathbf Q \mathbf H^\dag \mathbf H$ are small enough to approximate the following,

\begin{align}
\prod_{\ell=1}^{n_t} \left(1+\text{eig}_\ell \left( \mathbf Q \mathbf H^\dag \mathbf H \right) \right) \approx 1+\sum_{\ell=1}^{n_t} \text{eig}_\ell \left( \mathbf Q \mathbf H^\dag \mathbf H \right).
\end{align}
Therefore, the instantaneous mutual information of \cref{instantaneous-mutual-information-mimo-general} can be approximated by
\begin{align} \label{mimo_lsnr_mutual_approx}
\mathcal I &= \ln \det \left ( \mathbf I_{n_t}+ \mathbf Q \mathbf H^\dag \mathbf H \right) \nonumber \\
&=\ln \prod_{\ell=1}^{n_t} \left(1+\text{eig}_\ell \left( \mathbf Q \mathbf H^\dag \mathbf H \right) \right) \nonumber \\
&\approx \ln \left(1+\sum_{\ell=1}^{n_t} \text{eig}_\ell \left( \mathbf Q \mathbf H^\dag \mathbf H \right) \right).
\end{align}
Using \cref{mimo_lsnr_mutual_approx}, we can prove the following proposition on the optimum transmit covariance matrix which maximizes the throughput in the asymptotically low SNR regime MIMO channel.
%\begin{lemma}
%For every integer $n$ and every real $0 \leq s \leq 1$, $\frac{\Gamma \left( n,n s \right)}{(n-1)!}$ is a monotonically increasing function of $n$.
%\end{lemma}
%\begin{proof}
%
%\end{proof}
%\begin{lemma}
%For every integer $x$, $\Gamma \left( x, x \right)$ is a monotonically increasing function of $x$.
%\end{lemma}
%\begin{lemma}
%For every integer $l_t$ and $n_t$ ($l_t<n_t$) and every real $0 \leq s \leq n_r$, $\Gamma \left( n_t n_r, n_t s \right)>\Gamma \left( l_t n_r, l_t s \right)$
%\end{lemma}
%\begin{proof}
%
%\end{proof}

\begin{proposition} \label{mimo-lsnr_single-theorem}
The optimum transmit strategy maximizing the throughput in the asymptotically low SNR regime MIMO channel is transmitting independent signals and performing equal power allocation across all available antennas. The maximum throughput is 
\begin{align} \label{mimo_lsnr_thr_1}
\mathcal R_{s}^m = \max_{0 < s < n_r} \frac{\Gamma \left( n_t n_r, n_t s \right)}{(n_t n_r-1)!} \ln \left( 1+P s \right).
\end{align}
\end{proposition}

\begin{proof}
Let $\delta_\ell P$ denote the allocated power to the $\ell$'th antenna subject to $\sum_{\ell=1}^{n_t} \delta_{\ell} = 1$. 
%The instantaneous mutual information of \cref{instantaneous-mutual-information-mimo-general} can be approximated by
%If $n_t \leq n_r$, then
From \cref{mimo_lsnr_mutual_approx}, the instantaneous mutual information for low SNR values can be expressed as,
\begin{align} \label{mimo_lsnr_mutual}
\mathcal I & \approx \ln \left(1+\sum_{\ell=1}^{n_t} \text{eig}_\ell \left( \mathbf Q \mathbf H^\dag \mathbf H \right) \right) \nonumber \\
&= \ln \left(1+ \text{tr}\left( \mathbf Q \mathbf H^\dag \mathbf H \right) \right) \nonumber \\
&= \ln \left( 1+P \sum_{\ell=1}^{n_t} \sum_{k=1}^{n_r} \delta_\ell \left| h_{\ell,k} \right|^2  \right).
\end{align}
%Similarly, for $n_t \geq n_r$, one can write
%\begin{align} \label{mimo_lsnr_mutual_2}
%\mathcal I &= \ln \det \left ( \mathbf I_{n_r}+ \mathbf H \mathbf Q \mathbf H^\dag \right) \nonumber \\
%&=\ln \prod_{\ell=1}^{n_t} \left(1+\text{eig}_\ell \left( \mathbf H \mathbf Q \mathbf H^\dag \right) \right) \nonumber \\
%&\approx \ln \left(1+\sum_{\ell=1}^{n_t} \text{eig}_\ell \left( \mathbf H \mathbf Q \mathbf H^\dag \right) \right) \nonumber \\
%&= \ln \left(1+ \text{tr}\left( \mathbf H \mathbf Q \mathbf H^\dag \right) \right) \nonumber \\
%&= \ln \left( 1+P \sum_{\ell=1}^{n_t} \sum_{k=1}^{n_r} \delta_\ell \left| h_{\ell,k} \right|^2  \right).
%\end{align}

\Cref{mimo_lsnr_mutual} corresponds to the instantaneous mutual information in the MISO channel. Therefore, the optimum transmit strategy minimizing the outage probability in the asymptotically low SNR regime MIMO channel is to transmit independent signals and perform equal power allocation across a fraction of available antennas.

Assume that the transmitter has allocated equal power to $l_t$ out of $n_t$ transmit antennas. the maximum throughput is given by
\begin{align} \label{mimo_lsnr_thr_20}
\mathcal R_{s}^m = \max_{s} \frac{\Gamma \left( l_t n_r, l_t s \right)}{(l_t n_r-1)!} \ln \left( 1+P s \right)
% \nonumber \\
%&= \max_{\hat s} \frac{\Gamma \left( l_t n_r, l_t n_r \hat s \right)}{(l_t n_r-1)!} \ln \left( 1+P n_r \hat s \right)
.
\end{align}
With $\hat s = \frac{s}{n_r}$, \cref{mimo_lsnr_thr_20} leads to
\begin{align} \label{mimo_lsnr_thr_2}
\mathcal R_{s}^m 
%&= \max_{s} \frac{\Gamma \left( l_t n_r, l_t s \right)}{(l_t n_r-1)!} \ln \left( 1+P s \right) \nonumber \\
= \max_{\hat s} \frac{\Gamma \left( l_t n_r, l_t n_r \hat s \right)}{(l_t n_r-1)!} \ln \left( 1+P n_r \hat s \right).
\end{align}
\Cref{mimo_lsnr_thr_2} corresponds to the maximum throughput expression of the MISO channel, i.e., \cref{miso_single_Rayleigh3}, with $l_t n_r$ transmit antennas and total power $P n_r$. According to \cref{miso-single-theorem}, the optimum transmit strategy is to use all available antennas and $0 <\hat s< 1$, and equivalently $0 < s < n_r$.
%We show that the solution to \cref{mimo_lsnr_thr_1} is less than or equal to the number of receive antennas, i.e., $0 \leq s^o \leq n_r$.
%
%Following the steps in the proof of \cref{miso-single-theorem} and noting $\overline F_{a_{eq}}(l_t s) = \frac{\Gamma \left( l_t n_r, l_t s \right)}{(l_t-1)!}$, we have
%\begin{align} \label{lhs_diff_mimo_lsnr_single_Rayleigh2}
%r(s) = \frac{(l_t n_r-1)!e^{-l_t s}\sum_{k=0}^{l_t n_r-1} \frac{(l_t s)^k}{k!}}{l_t (l_t s)^{l_t n_r-1}e^{-l_t s}} \nonumber \\
%=\frac{1}{l_t}+ \frac{1}{l_t} \sum_{k=0}^{l_t n_r-2} \frac{(l_t n_r-1)\dots(k+1)}{(l_t s)^{l_t n_r-k-1}} \nonumber \\
%=
%\frac{1}{l_t}+\frac{1}{l_t} \sum_{k=0}^{l_t n_r-2} \prod_{i=0}^{l_t n_r-k-2} \frac{l_t-i-1}{l_t s}.
%\end{align}
%For $s \geq s_c = 1$ we have $\frac{l_t-i-1}{l_t s} \leq 1$. Replacing in (\ref{lhs_diff_miso_single_Rayleigh2}) gives
%\begin{align} \label{lhs_diff_miso_single_Rayleigh3}
%r(s) < \frac{1}{l_t}+\frac{1}{l_t} \sum_{k=0}^{l_t-2} \prod_{i=0}^{l_t-k-2} 1 =
%\frac{1}{l_t}+\frac{l_t-1}{l_t}= 1, \forall s \geq 1.
%\end{align}
%According to equation (\ref{lhs_diff_miso_single_Rayleigh2}), $\lim_{s\to 0} r(s) = +\infty$.

\end{proof}

In the same direction, the finite-layer expected-rate is given by \cref{mimo-lsnr-multi-theorem}.

\begin{corollary} \label{mimo-lsnr-multi-theorem}
The optimum transmit strategy maximizing the $K$-layer expected-rate of the asymptotically low SNR regime MIMO channel is transmitting independent signals and performing equal power allocation across all available antennas in each code layer. The maximum throughput is 
\begin{align} \label{mimo_lsnr_thr_multi-1}
\mathcal R_{f}^m {=} \!\max_{\begin{subarray}{c}
0 < s_i < n_r, P_i \\
%0 \leq P_l \leq P \\
\sum_{i=1}^K P_i = P
\end{subarray} }
\sum_{i=1}^K \frac{\Gamma \left( n_t n_r, n_t s \right)}{(n_t n_r-1)!} \!\ln \!\left(\! 1\!+\!\frac{P_i s_i}{1\!+\!\sum_{j=i+1}^K P_j s_i} \!\right)\!.
\end{align}
\end{corollary}

\begin{proof}
At the $i$'th layer, let $\delta_\ell P_i$ and $\eta_\ell I_i$ denote the allocated power and upper-layers power at the $\ell$'th antenna subject to $\sum_{\ell=1}^{n_t} \delta_{\ell} = \sum_{\ell=1}^{n_t} \eta_{\ell} = 1$, and $I_i = \sum_{j=i+1}^{K} P_j$. 
Following the same steps in \cref{mimo_lsnr_mutual}, the $i$'th layer instantaneous mutual information can be approximated by
%If $n_t \leq n_r$, then
\begin{align} \label{mimo_lsnr_mutual_ml}
\mathcal I_i \approx \ln \left(1+\frac{P_i \sum_{\ell=1}^{n_t} \sum_{k=1}^{n_r} \delta_\ell \left| h_{\ell,k} \right|^2}{1+I_i \sum_{\ell=1}^{n_t} \sum_{k=1}^{n_r} \eta_\ell \left| h_{\ell,k} \right|^2} \right).
\end{align}
\Cref{mimo_lsnr_mutual_ml} corresponds to the instantaneous mutual information of the multi-layer MISO channel in \cref{miso-multi_finite-subsection}. The proof is completed by following the steps in the proof of \cref{miso-multi-theorem} and \cref{mimo-lsnr_single-theorem}.
\end{proof}

Corresponding to \cref{miso-infinite-proposition}, we have the following corollary for continuous-layer coding in the low SNR MIMO channels.
\begin{corollary}
The maximum continuous-layer expected-rate in the asymptotically low SNR regime MIMO channel is given by 
\begin{align}
\mathcal R_{c}^m = \mathcal R(s_1) - \mathcal R(s_0),
\end{align}
where,
%\begin{align}
%g(x) = e^{-x} \sum_{\ell = 0}^{n_t-1} \frac{x^{\ell}}{\ell !} - e^{-x} \sum_{\ell = 1}^{n_t-1} \frac{n_t+ 1-\ell}{\ell} \sum_{i = 0}^{\ell-1} \frac{x^i}{i!} \nonumber \\
%-(n_t+1) \emph{E}_1 (x) 
%\end{align}
\begin{align} \label{indefinite-integral-mimo-lsnr-cl}
\mathcal R(s) {=} e^{-s} \! \sum_{\ell = 1}^{n_t n_r-1} \frac{1}{\ell !}\! \left(\! s^{\ell} \!-\! (n_t n_r\!+\! 1\!-\!\ell)(\ell \!-\!1)! \sum_{k = 0}^{\ell-1} \frac{s^k}{k!}  \!\right) \nonumber \\
+e^{-s}\!-\!(n_t n_r\!+\!1) \emph{E}_1 (s).
\end{align}
$s_0$ and $s_1$ are the solutions to 
\begin{align}
\begin{cases}
\sum_{\ell=0}^{n_t n_r-1} \frac{(n_t n_r-1)!}{\ell ! s_0^{n_t n_r-\ell}} = 1+\frac{P}{n_t} s_0, \\
\sum_{\ell=0}^{n_t n_r-1} \frac{(n_t n_r-1)!}{\ell ! s_1^{n_t n_r-\ell}} = 1,
\end{cases} 
\end{align}
respectively.
\end{corollary}

%\begin{remark}
%Since in asymptotically low SNR regime, the outage probability is Schur-concave with respect to the channel covariance matrix \cite{jorswieck}, the maximum throughput is a Schur-convex function of the channel covariance matrix, that is channel correlation increases the maximum throughput.
%\end{remark}

\begin{remark}
Analogous to the MISO channel, in the asymptotically low SNR regime MIMO channel with uninformed transmitter, 
%the maximum throughput and maximum expected-rate are Schur-concave functions of the channel covariance matrix, that is 
channel correlation decreases the maximum throughput and maximum expected-rate.
\end{remark}

%%%%%%%%%%%%%%%%%%%%%%%%%%%%%%

\subsection{Asymptotically High SNR Regime}

For large SNR values, we take advantages of Wishart distribution properties.
% and also the fact that the eigenvalues of $\mathbf Q \mathbf H^{\dag} \mathbf H $ and $\mathbf H \mathbf Q \mathbf H^{\dag}$ are large. 
In order to enhance the lucidity of this section, let us define $p \Def \min \left\{ n_t , n_r \right\}$, $n \Def \max \left\{ n_t , n_r \right\}$, and
\begin{align}
\mathbf W = 
\begin{cases}
\mathbf H^{\dag} \mathbf H & n_t \leq n_r, \\
\mathbf H \mathbf H^{\dag} & n_t > n_r.
\end{cases}
\end{align}
Matrix $\mathbf W$ has a central complex $p$-variate Wishart distribution with scale matrix $\mathbf I_{n_t}$ and $n$ degrees of freedom \cite{anderson2003,muirhead1982aspects,ratnarajah2005complex}.

\Cref{mimo_hsnr_thr} 
%provides the optimum transmission scheme for high SNR MIMO channels. 
yields the maximum throughput in the asymptotically high SNR regime MIMO channel by obtaining the optimum transmit covariance matrix $\mathbf Q^o$.

\begin{theorem} \label{mimo_hsnr_thr}
The optimum transmit strategy maximizing the throughput in the asymptotically high SNR regime MIMO channel is sending independent signals and performing equal power allocation across all available antennas. The maximum throughput is 
\begin{align} \label{mimo_hsnr_single_rate}
\mathcal R_{s}^m &= \max_{s}
\overline F_{\mathrm a}\left( \frac{n_{t}^{p} s}{P^{p-1}} \right) \ln \left( 1+P s \right) \\
&=\max_z \mathcal Q(z)\! \left(\! z \sqrt{\frac{\pi^2}{6} p - \sum_{k=0}^{p-1} \sum_{\ell=1}^{n-k-1} \frac{1}{\ell^2}}  \right. \nonumber \\
&+  p\left( \digamma(1) +\ln \left( \frac{P}{n_t} \right) \right)+ \sum_{k=0}^{p-1} \sum_{\ell=1}^{n-k-1} \frac{1}{\ell} \!\Bigg),
\end{align}
where
$-\digamma(1) \approx 0.577215$ is the E$\ddot{\text{u}}$ler-Mascheroni constant, 
$a \Def \prod_{\ell=1}^{p} a_{\ell,\ell}^2$, and 
$a_{\ell,\ell}^2,\forall \ell$ are independent gamma-distributed with scale 1 and shape $n-\ell+1$, i.e.,
$f_{\mathrm a_{\ell,\ell}^2}(x)=\frac{\Gamma \left( n-\ell+1 , x\right)}{\left( n-\ell \right)!}$.
%$a_{\ell,\ell}^2 \sim \Gamma \left(n-\ell+1,1 \right)$ independently.
\end{theorem}

\begin{proof}
Again, we first assume that $l_t$ out of $n_t$ antennas are active. Then, we shall see that the optimum solution is $l_t^o=n_t$.
Define the index set $Z\left(\mathbf Q \right)\Def \left\{ \ell:q_{\ell,\ell}=0 \right\}$. Denote by $\mathbf Q_{l_t}$ the matrix obtained from $\mathbf Q$ by eliminating of all the $\ell$'th rows and columns with $\ell \in Z\left(\mathbf Q \right)$.
%Let us denote $\mathbf Q_{l_t}$ as a sub-matrix of $\mathbf Q$ whose diagonal components are the large value (non-zero) diagonal entries of $\mathbf Q$. 
Clearly, $\mathbf Q_{l_t}$ has full rank.
We divide the proof into two parts: Part i) $l_t \leq n_r$, Part ii) $l_t \geq n_r$. We wish to show that in both cases, the throughput is a strictly increasing function with respect to $l_t$.

\textbf{Part i)}:
%If $n_t \leq n_r$, then

In high SNR regime, the eigenvalues of $\mathbf Q_{l_t} \mathbf H^\dag \mathbf H$ are large. The instantanous mutual information can be well approximated by
\begin{align} \label{mimo_hsnr_mutual}
\mathcal I &= \ln \det \left ( \mathbf I_{l_t}+ \mathbf Q \mathbf H^\dag \mathbf H \right) \nonumber \\
&= \ln \prod_{\ell = 1}^{l_t} \left(1+ \text{eig}_{\ell} \left( \mathbf Q \mathbf H^{\dag} \mathbf H \right) \right) \nonumber \\
&\approx \ln \prod_{\ell = 1}^{l_t} \left( \text{eig}_{\ell} \left( \mathbf Q_{l_t} \mathbf H^{\dag} \mathbf H \right) \right)  \nonumber \\
&= \ln \det \left( \mathbf Q_{l_t} \mathbf H^{\dag} \mathbf H  \right) \nonumber \\
&= \ln \det \mathbf Q_{l_t} + \ln \det \left( \mathbf H^{\dag} \mathbf H  \right) \nonumber \\
&= \ln \det \mathbf Q_{l_t}  + \ln \det \mathbf W. 
\end{align}
%\Cref{mimo_hsnr_mutual_1} and \cref{mimo_hsnr_mutual_2} allow the instantaneous mutual information of \cref{instantaneous-mutual-information-mimo-general} in asymptotically high SNR values to be approximated by
%\begin{align} \label{mimo_hsnr_mutual}
%\mathcal I = \ln \det \left( \mathbf Q \right) + \ln \det \left( \mathbf W  \right). 
%\end{align}

Clearly, the CDF of $\ln \det \mathbf W$ decreases by the use of more antennas.
%, because it can be approximated by a Gaussian distribution whose mean is raised using more  . 
We shall now show that $\ln \det \mathbf Q_{l_t}$ and thereby, $\mathcal I$ increases with the number of active antennas. 
It is straight forward to verify that the solution to the maximization problem $\max \mathbf \det \mathbf Q_{l_t}$ subject to ${\text{tr}\left(\mathbf Q_{l_t}\right)=P}$ over diagonal matrices is $\mathbf Q_{l_t}=\frac{P}{l_t} \mathbf I_{l_t}$. Therfore, \cref{mimo_hsnr_mutual} is simplified as follows
\begin{align}
\mathcal I \approx l_t \ln \left( \frac{P}{l_t} \right) + \ln \det \mathbf W.
\end{align}
For $P>e l_t$,
\begin{align}
\frac{\partial \mathcal I}{\partial l_t} = \ln \left( \frac{P}{l_t} \right) - 1 > 0.
\end{align}
As a result, in high SNR regime, the instantaneous mutual information $\mathcal I$ strict monotonic increasing with respect to the number of transmit antennas.

{\bf Part ii):}

In this case, we approximate the instantaneous mutual information as follows.
%Similarly, for $n_t \geq n_r$, one can write
\begin{align} \label{mimo_hsnr_mutual_2}
\mathcal I &= \ln \det \left ( \mathbf I_{n_r}+ \mathbf H \mathbf Q \mathbf H^\dag \right) \nonumber \\
&= \ln \prod_{\ell = 1}^{n_r} \left(1+ \text{eig}_{\ell} \left( \mathbf H \mathbf Q \mathbf H^{\dag} \right) \right) \nonumber \\
&\approx \ln \prod_{\ell = 1}^{n_r} \left( \text{eig}_{\ell} \left(\mathbf H \mathbf Q_{l_t} \mathbf H^{\dag} \right) \right)  \nonumber \\
&= \ln \det \left(\mathbf H \mathbf Q_{l_t} \mathbf H^{\dag}  \right). 
\end{align}

In this case, let us assume that the transmitter performs equal power allocation. Therefore,
\begin{align} \label{mimo_hsnr_mutual_3}
\mathcal I &\approx n_r \ln \left(\frac{P}{l_t}\right)+ \ln \det \left(\mathbf H \mathbf H^{\dag}  \right) \nonumber \\
&=  n_r \ln \left(\frac{P}{l_t}\right)+ \ln \det \mathbf W. 
\end{align}
In the following, we shall establish that the maximum throughput of the channel is strictly increasing with respect to $l_t$. 
From the maximization problem of \cref{throughput_gaussian_approx}, the maximum throughput can be equivalently expressed as
\begin{align} \label{throughput_gaussian_approx_hsnr}
\mathcal R_{s}^m = \max_z \mathcal Q(z) \left( \sigma(l_t,n_r) z + \mu(l_t,n_r) \right),
%\max_z \mathcal Q(z) \Bigg( \text{Var} \left( \ln \det \mathbf W  \right) z + \mathbb E \left( \ln \det \mathbf W  \right)  \nonumber \\
%+ n_r \ln \left(\frac{P}{n_t}\right) \Bigg)
\end{align}
with
\begin{align} 
&\mu(l_t,n_r) = \mathbb E \left( \ln \det \mathbf W  \right)+ p \ln \left(\frac{P}{l_t}\right), \label{mu_gaussian_approx_hsnr} \\
&\sigma^2(l_t,n_r)= \text{Var} \left( \ln \det \mathbf W  \right). \label{sigma_gaussian_approx_hsnr}
\end{align}
A central complex Wishart-distributed matrix $\mathbf W$ satisfies \cite{verdu_book}
\begin{align}
&\mathbb E \left( \ln \det \mathbf W  \right) = \sum_{k=0}^{p-1} \digamma(n-k), \label{e_wishart_hsnr} \\
&\text{Var} \left( \ln \det \mathbf W  \right) = \sum_{k=0}^{p-1} \digamma^{\prime}(n-k). \label{var_wishart_hsnr}
\end{align}
For natural arguments, the E$\ddot{\text{u}}$ler's digamma function and its derivative, i.e., $\digamma(m)$ and $\digamma^{\prime}(m)$, can be expressed as
\begin{align}
&\digamma(m) = \digamma(1)+\sum_{\ell=1}^{m-1} \frac{1}{\ell}, \label{self_psi} \\
&\digamma^{\prime}(m) = \frac{\pi^2}{6} - \sum_{\ell=1}^{m-1} \frac{1}{\ell^2}, \label{dot_psi}
\end{align}
with $-\digamma(1)=-\Gamma^{\prime}(1)=\lim_{m\to \infty} \left( \sum_{\ell=1}^m \frac{1}{\ell} - \ln(m) \right) \approx0.577215$ the E$\ddot{\text{u}}$ler-Mascheroni constant.
Inserting \cref{dot_psi} into \cref{var_wishart_hsnr} and then into \cref{sigma_gaussian_approx_hsnr} to obtain
\begin{align} \label{sigma_hsnr_simplified}
\sigma^2(l_t,n_r) = \frac{\pi^2}{6} n_r - \sum_{k=0}^{n_r-1} \sum_{\ell=1}^{l_t-k-1} \frac{1}{\ell^2},
\end{align}
we see that $\sigma^2(l_t,n_r)$ is a monotonically decreasing function with respect to $l_t$. 
Whereas $\mu(l_t,n_r)$ is a strictly increasing function with respect to both $l_t$ and $n_r$ as it represents the ergodic capacity of the high SNR $l_t \times n_r$ MIMO channel.
On the other hand, $\sigma^2(l_t,n_r)=\sum_{k=0}^{p-1} \digamma^{\prime}(n-k)$ is a monotonically increasing function with respect to $n_r$, because of the Basel problem, i.e., $\lim_{m\to \infty} \sum_{\ell=1}^{m} \frac{1}{\ell^2}= \frac{\pi^2}{6}$, which verifies that $\digamma^{\prime}(m) \geq 0$.

%Dividing \cref{mu_gaussian_approx_hsnr} by \cref{sigma_gaussian_approx_hsnr}, we get
%\begin{align}
%\frac{\mu(n_t,n_r)}{\sigma(n_t,n_r)} &= \frac{\digamma(1) n_r + \sum_{k=0}^{n_r-1} \sum_{\ell=1}^{n_t-k-1} \frac{1}{\ell}+ n_r \ln \left(\frac{P}{n_t}\right)}{\sqrt{\frac{\pi^2}{6} n_r - \sum_{k=0}^{n_r-1} \sum_{\ell=1}^{n_t-k-1} \frac{1}{\ell^2}}} \nonumber \\
%&\geq \frac{\digamma(1) n_r + \sum_{k=0}^{n_r-1} \sum_{\ell=1}^{n_t-k-1} \frac{1}{\ell}}{\sqrt{\frac{\pi^2}{6} n_r - \sum_{k=0}^{n_r-1} \sum_{\ell=1}^{n_t-k-1} \frac{1}{\ell^2}}}
%\end{align}

As the $\mathcal Q$-function is upper-bounded by the Chernoff bound, i.e., $\mathcal Q(z) \leq \frac{1}{2} e^{-\frac{z^2}{2}},~z \geq 0$, we have for $z \geq 0$,
\begin{align} \label{large_M_z_positive_ineq_hsnr}
-\frac{1}{\sqrt{2\pi}} e^{-\frac{z^2}{2}} \left( \sigma(l_t,n_r) z + \mu(l_t,n_r) \right) + \sigma(l_t,n_r) \mathcal Q(z) \nonumber \\ \leq
-\frac{1}{\sqrt{2\pi}} e^{-\frac{z^2}{2}} \sigma(l_t,n_r) \left( z + \frac{\mu(l_t,n_r)}{\sigma(l_t,n_r)} - \sqrt{\frac{\pi}{2}} \right) \stackrel{(a)}{<} 0,
\end{align} 
where $(a)$ follows the fact that $z \geq 0$ and $\frac{\mu(l_t,n_r)}{\sigma(l_t,n_r)} - \sqrt{\frac{\pi}{2}}>0$ as $P$ and thereby $\mu(l_t,n_r)$ is large.
From \cref{large_M_zo_eq,large_M_z_positive_ineq_hsnr}, one immediately finds that $z^o<0$.
Recall from \cref{throughput_gaussian_approx_a}, the maximum throughput is a strictly increasing function with respect to $l_t$ because $\mathcal R_{s}^m$ is a strictly increasing function with respect to $\mu(l_t,n_r)$, a monotonically decreasing function with respect to $\sigma(l_t,n_r)$, and $z^o<0$.

Thus, in both parts, i.e., $l_t \leq n_r$ and $l_t \geq n_r$, $\mathcal R_{s}^m$ is a strictly increasing function with respect to $l_t$. We conclude that in the asymptotically high SNR regime MIMO channel, the maximum throughput is a strictly increasing function with respect to the number of active transmit antennas, and hence, $l_t^o=n_t$.

%Let us define $p = \min \left\{ n_t , n_r \right\}$, $n = \max \left\{ n_t , n_r \right\}$, and
%\begin{align}
%\mathbf W = 
%\begin{cases}
%\mathbf H^{\dag} \mathbf H & n_t \leq n_r \\
%\mathbf H \mathbf H^{\dag} & n_t > n_r.
%\end{cases}
%\end{align}
%Matrix $\mathbf W$ has a central complex $p$-variate Wishart distribution with scale matrix $\mathbf I_{n_t}$ and $n$ degrees of freedom \cite{anderson2003}. 
Performing Bartlett decomposition \cite{kshirsagar1959bartlett}, we get $\mathbf W=\mathbf A \mathbf A^\dag$, where $\mathbf A$ is a square lower triangular matrix (left triangular matrix) in the form of
\begin{align}
\mathbf A =
\begin{bmatrix}
a_{1,1} & 0 & 0 & \cdots & 0 \\
a_{2,1} & a_{2,2} & 0 & \cdots & 0 \\
a_{3,1} & a_{3,2} & a_{3,3} & \cdots & 0 \\
\vdots & \vdots & \vdots & \ddots & \vdots \\
a_{p,1} & a_{p,2} & a_{p,3} & \cdots & a_{p,p}
\end{bmatrix},
\end{align}
where $a_{\ell,k} \sim \mathcal{CN}(0,1), \ell \neq k$, and $a_{\ell,\ell}^2,\forall \ell$ are independent gamma-distributed with scale 1 and shape $n-\ell+1$.
Clearly, $\det \mathbf W=\det \mathbf A \times \det \mathbf A^\dag = \prod_{\ell=1}^{p} a_{\ell,\ell}^2$.

%Since $\det \left(\mathbf I_{n_t}+\frac{P}{n_t} \mathbf H^{\dag} \mathbf H \right)=\det \left(\mathbf I_{n_r}+\frac{P}{n_t} \mathbf H \mathbf H^{\dag} \right)$, one can write 
%\begin{align}
%\mathcal I \approx n_{min} \ln \left( \frac{P}{n_t} \right) + \ln \det \left( \mathbf W  \right).
%\end{align}
Therefore, the maximum throughput is 
\begin{align} \label{mimo_hsnr_thr_1}
\mathcal R_{s}^m &= \max_{s} \Pr \left\{ \det \left( \frac{P}{n_t} \mathbf W \right) \geq P s \right\} \ln \left( 1+P s \right) \nonumber \\
&=\max_{s} \Pr \left\{ \det \mathbf W \geq \frac{n_{t}^{p} s}{P^{p-1}} \right\} \ln \left( 1+P s \right) \nonumber \\
&= \max_{s} \Pr \left\{ \prod_{\ell=1}^{p} a_{\ell,\ell}^2 \geq \frac{n_{t}^{p} s}{P^{p-1}} \right\} \ln \left( 1+P s \right).
\end{align}

%Therefore, 
%\begin{align} \label{mimo_hsnr_thr_1}
%\mathcal R_{s}^m = \max_{s} \Pr \left\{ \prod_{\ell=1}^{p} a_{\ell,\ell}^2 \geq \frac{n_{t}^{p} s}{P^{p-1}} \right\} \ln \left( 1+P s \right).
%\end{align}

%Since $\mathcal I$ and consequently $\ln \det \mathbf W$ converges in distribution to a Gaussian distribution, we have $\ln \det \mathbf W \sim \mathcal N \left( \mu , \sigma^2 \right)$, where
%\begin{align}
%\begin{cases}
%\mu = \mathbb E \left( \ln \det \mathbf W  \right), \\
%\sigma^2 = \text{Var} \left( \ln \det \mathbf W  \right).
%\end{cases}
%\end{align}
%From \cref{throughput_gaussian_approx}, the throughput can also be written as 
%\begin{align}
%\mathcal R_{s}^m 
%%&= \max_R \mathcal Q \left( \frac{R-\mu}{\sigma} \right) R 
%&= \max_z \mathcal Q(z) \left( \sigma\left( n_t,n_r \right) z + \mu\left( n_t,n_r \right) \right) 
%%\nonumber \\
%%&=\max_z \mathcal Q(z)\! \left(\! \sqrt{\sum_{\ell=0}^{p-1} \digamma^{\prime}(n-\ell)} z +\!\!\! \sum_{\ell=0}^{p-1} \digamma(n-\ell) \!\right)
%.
%%\\
%%&=\max_z \mathcal Q(z)\! \left(\! z \sqrt{\frac{\pi^2}{6} p - \sum_{k=0}^{p-1} \sum_{\ell=1}^{n-k-1} \frac{1}{\ell^2}}  \right.\nonumber \\
%%&+\!\!\! p\left( \digamma(1) +\ln \left( \frac{P}{n_t} \right) \right)+ \sum_{k=0}^{p-1} \sum_{\ell=1}^{n-k-1} \frac{1}{\ell} \!\Bigg)
%\end{align}
From \cref{throughput_gaussian_approx_hsnr,mu_gaussian_approx_hsnr,sigma_gaussian_approx_hsnr,e_wishart_hsnr,var_wishart_hsnr,self_psi,dot_psi,sigma_hsnr_simplified}, the throughput can also be written as 
\begin{align}
\mathcal R_{s}^m &= \max_z \mathcal Q(z) \left( \sigma\left( n_t,n_r \right) z + \mu\left( n_t,n_r \right) \right) \nonumber \\
&=\max_z \mathcal Q(z) \Bigg(z \sqrt{\sum_{\ell=0}^{p-1} \digamma^{\prime}(n-\ell)} \nonumber \\
&+ p \ln \left( \frac{P}{n_t} \right) + \sum_{\ell=0}^{p-1} \digamma(n-\ell) \Bigg) \nonumber \\
&= \max_z \mathcal Q(z)\! \Bigg(\! z \sqrt{\frac{\pi^2}{6} p - \sum_{k=0}^{p-1} \sum_{\ell=1}^{n-k-1} \frac{1}{\ell^2}}  \nonumber \\
&+ p\left( \digamma(1) +\ln \left( \frac{P}{n_t} \right) \right)+ \sum_{k=0}^{p-1} \sum_{\ell=1}^{n-k-1} \frac{1}{\ell} \!\Bigg).
\end{align}

\end{proof}

\begin{remark}
Since in asymptotically high SNR regime, the outage probability is Schur-convex with respect to the channel covariance matrix \cite{jorswieck}, the maximum throughput is a Schur-concave function of the channel covariance matrix, i.e., channel correlation decreases the maximum throughput.
\end{remark}

%%%%%%%%%%%%%%%%%%%%%%%%%%%%%%

\subsection{Asymptotically Large Number of Antennas}

Here, two asymptotic results for large number of transmit antennas and large number of receive antennas are presented.
As pointed out earlier, we can restrict our attention to diagonal transmit covariance matrices.
%In the same general outline to the previous sections, 
To prove by contradiction, first we assume that the optimum transmit covariance matrix is $\mathbf Q^o = \frac{P}{l_t} \mathbf I_{l_t}$; next, we shall show that the maximum throughput increases with the number of transmit antennas and hence,  
%larger number of transmit antennas, the higher throughput, and as a result, 
$\mathbf Q^o = \frac{P}{n_t} \mathbf I_{n_t}$. Finally, we formulate the maximum throughput.

In following, \cref{theorem-high-transmit-antenna,theorem-high-receive-antenna} yield the maximum throughput of asymptotically large number of transmit antennas and asymptotically large number of receive antennas, respectively.
In the proof of both theorems, we use the results presented by Hochwald, Marzetta, and Tarokh \cite{hochwald} which provide us with approximations for mean and variance of the instantaneous mutual information in the large number of transmit antennas and large number of receive antennas asymptotes. 
%ere derived as follows
%\begin{align} 
%&n_t \to \infty
%\begin{cases}
%\mu_{l_t} = n_r \ln \left( 1+P \right), \\
%\sigma_{l_t}^2 = \frac{n_r P^2}{l_t(1+P^2)},
%\end{cases} \label{large_M_stats} \\
%&n_r \to \infty
%\begin{cases}
%\mu_{l_t} = l_t \ln \left( 1+\frac{n_r}{l_t}P \right), \\
%\sigma_{l_t}^2 = \frac{l_t}{n_r}.
%\end{cases} \label{large_N_stats}
%\end{align}

%\subsubsection{asymptotically large Number of Transmit Antennas}
\begin{theorem} \label{theorem-high-transmit-antenna}
In the MIMO channel with asymptotically large number of transmit antennas, the optimum transmit covariance matrix which maximizes the throughput is $\mathbf Q^o = \frac{P}{n_t} \mathbf I_{n_t}$. The maximum throughput of the channel is given by
\begin{align}
\mathcal R_{s}^m = \max_z \mathcal Q(z) \left( \sqrt{\frac{n_r}{n_t}} \frac{P}{\sqrt{1+P^2}} z + n_r \ln \left( 1+P \right) \right).
\end{align}
\end{theorem}

\begin{proof}
%Thanks to the results of \cite{hochwald}
According to the results provided in \cite{hochwald}, we have
\begin{align} 
\begin{cases}
\lim_{n_t \to \infty} \mu \left(l_t,n_r \right) = n_r \ln \left( 1+P \right), \\
\lim_{n_t \to \infty} \sigma^2 \left(l_t,n_r \right) = \frac{n_r P^2}{l_t(1+P^2)}.
\end{cases} \label{large_M_stats} 
\end{align}
From \cref{large_M_stats} and noting the $\mathcal Q$-function's Chernoff bound, i.e., $\mathcal Q(z) \leq \frac{1}{2} e^{-\frac{z^2}{2}},~ z \geq 0$, we have for $z \geq 0$,
\begin{align} \label{large_M_z_positive_ineq}
-\frac{1}{\sqrt{2\pi}} e^{-\frac{z^2}{2}} \left( \sigma\left(l_t,n_r \right) z + \mu\left(l_t,n_r \right) \right) + \sigma\left(l_t,n_r \right) \mathcal Q(z) \nonumber \\ \leq 
-\frac{1}{\sqrt{2\pi}} e^{-\frac{z^2}{2}} \sigma\left(l_t,n_r \right) \left( z + \frac{\mu\left(l_t,n_r \right)}{\sigma\left(l_t,n_r \right)} - \sqrt{\frac{\pi}{2}} \right) \stackrel{(a)}{<} 0,
\end{align} 
where (a) comes from the fact that for $z \geq 0$,
\begin{align}
&z + \frac{\mu\left(l_t,n_r \right)}{\sigma\left(l_t,n_r \right)} - \sqrt{\frac{\pi}{2}} \geq \frac{\mu\left(l_t,n_r \right)}{\sigma\left(l_t,n_r \right)}- \sqrt{\frac{\pi}{2}} \nonumber \\
&= \sqrt{n_r l_t}\sqrt{1+\frac{1}{P^2}}\ln(1+P)- \sqrt{\frac{\pi}{2}} \stackrel{l_t \to \infty}{>} 0 .
\end{align}
Comparing \cref{large_M_zo_eq,large_M_z_positive_ineq}, we have $z^o < 0$. Since $\mu\left(l_t,n_r \right)$ does not depend on $l_t$, $\sigma\left(l_t,n_r \right)$ is a strictly decreasing functions with respect to $l_t$, and $z^o < 0$, one can conclude that $\mathcal R_{s}^m=\mathcal Q(z^o) \left( \sigma\left(l_t,n_r \right) z^o + \mu\left(l_t,n_r \right) \right)$ is a strictly increasing function with respect to $l_t$. Thus, $\mathbf Q^o = \frac{P}{n_t} \mathbf I_{n_t}$.

\end{proof}

\begin{theorem} \label{theorem-high-receive-antenna}
In the MIMO channel with asymptotically large number of receive antennas, the optimum transmit covariance matrix which maximizes the throughput is $\mathbf Q^o = \frac{P}{n_t} \mathbf I_{n_t}$. The maximum throughput of the channel is given by 
\begin{align}
\mathcal R_{s}^m = \max_z \mathcal Q(z) \left( \sqrt{\frac{n_t}{n_r}} z + n_t \ln \left( 1+\frac{n_r}{n_t}P \right) \right).
\end{align}
\end{theorem}

\begin{proof}
As the number of receive antennas goes to infinity, the mean and variance of the channel mutual information obey \cite{hochwald}
\begin{align} 
\begin{cases}
\lim_{n_r \to \infty} \mu\left(l_t,n_r \right) = l_t \ln \left( 1+\frac{n_r}{l_t}P \right), \\
\lim_{n_r \to \infty} \sigma^2\left(l_t,n_r \right) = \frac{l_t}{n_r}.
\end{cases} \label{large_N_stats}
\end{align}

From \cref{throughput_gaussian_approx,large_M_stats}, the maximum throughput is 
\begin{align} \label{large_N_proof_eq}
\mathcal R_{s}^m &= \max_z \mathcal Q(z) \left( \sqrt{\frac{l_t}{n_r}} z + l_t \ln \left( 1+\frac{n_r}{l_t}P \right) \right) \nonumber \\
&\stackrel{(a)}{\geq} \mathcal Q(-\sqrt{n_r}) \left( -\sqrt{l_t}  + l_t \ln \left( 1+\frac{n_r}{l_t}P \right) \right) \nonumber \\
&\stackrel{(b)}{>} \mathcal Q(-\sqrt{n_r}) \left( -\ln\left( 1+\frac{n_r P}{l_t-1} \right) \right. \nonumber \\ 
&\qquad \qquad \left. - l_t \ln \left( 1-\frac{1}{l_t} \right) + l_t \ln \left( 1+\frac{n_r}{l_t}P \right) \right) \nonumber \\
&\stackrel{(c)}{>} \mathcal Q(-\sqrt{n_r}) \left( (l_t-1) \ln \left( 1+\frac{n_r}{l_t-1}P \right) \right)  \nonumber \\
%&=\mathcal Q(-\sqrt{n_r}) \left( (l_t-1) \ln \left( 1+\frac{n_r}{l_t-1}P \right) \right) \nonumber \\
&\stackrel{(d)}{\geq} \left( 1-\frac{1}{2}e^{-\frac{n_r}{2}} \right) \left( (l_t-1) \ln \left( 1+\frac{n_r}{l_t-1}P \right) \right) \nonumber \\
&\stackrel{(e)}{\stackrel{n_r \to \infty}{\longrightarrow}} (l_t-1) \ln \left( 1+\frac{n_r}{l_t-1}P \right) \nonumber \\
&\stackrel{(f)}{\geq} \!\max_z \mathcal Q(z)\! \left(\! \sqrt{\frac{l_t\!-\!1}{n_r}} z \!+\! (l_t\!-\!1) \ln\! \left(\! 1\!+\!\frac{n_r}{l_t\!-\!1}P \right) \!\right),
\end{align}
where (a) follows from choosing $z=-\sqrt{n_r}$ instead of its optimum value, (b) follows form $\sqrt{l_t}+l_t\ln \left( \frac{l_t}{l_t-1} \right) < \ln \left( 1+\frac{n_r P}{l_t-1} \right)$ for large values of $n_r$, (c) follows from algebraic simplifications, (d) follows from the $\mathcal Q$-function's Chernoff bound, (e) follows from $\lim_{n_r \to \infty} e^{-\frac{n_r}{2}} \ln \left( 1+\frac{n_r}{l_t-1}P \right) = 0$, and (f) follows from the fact that the maximum throughput is always less than or equal to the ergodic capacity based on \cref{Markov-throughput-ergodic-lemma}.

\Cref{large_N_proof_eq} proves that $\mathcal R_{s}^m$ is a strictly increasing function with respect to $l_t$, and hence, $\mathbf Q^o = \frac{P}{n_t} \mathbf I_{n_t}$.
\end{proof}

%\subsection{Expected-rate of MIMO Channels}

%%%%%%%%%%%%%%%%%%%%%%%%%%%%%%%%%%%%%%%%%%%%%%%%%%%%%%%%%%%%%%%%%%%%%%%%%%%%%%%%%%%%%%%%%%%%%%%%%%%%%%%%%%%%%%%%%%%
%%%%%%%%%%%%%%%%%%%%%%%%%%%%%%%%%%%%%%%%%%%%%%%%%%%%%%%%%%%%%%%%%%%%%%%%%%%%%%%%%%%%%%%%%%%%%%%%%%%%%%%%%%%%%%%%%%%
%%%%%%%%%%%%%%%%%%%%%%%%%%%%%%%%%%%%%%%%%%%%%%%%%%%%%%%%%%%%%%%%%%%%%%%%%%%%%%%%%%%%%%%%%%%%%%%%%%%%%%%%%%%%%%%%%%%
%%%%%%%%%%%%%%%%%%%%%%%%%%%%%%%%%%%%%%%%%%%%%%%%%%%%%%%%%%%%%%%%%%%%%%%%%%%%%%%%%%%%%%%%%%%%%%%%%%%%%%%%%%%%%%%%%%%
%%%%%%%%%%%%%%%%%%%%%%%%%%%%%%%%%%%%%%%%%%%%%%%%%%%%%%%%%%%%%%%%%%%%%%%%%%%%%%%%%%%%%%%%%%%%%%%%%%%%%%%%%%%%%%%%%%%

\section{Two-Transmitter Distributed Antenna Systems} \label{distributed-section}

%Since the base stations are typically not mobile, they can communicate through a high-speed reliable connection \cite{goldsmith}. 
There has been some research in assumption of perfect cooperation between base stations, and consequently treat them as distributed antennas of one base station \cite{goldsmith}. % \cite{jafar_cellular_conf, shamai_cellular_conf, jafar_cellular_conf2}. 
%Since power cooperation and also perfect cooperation is not of practical relevance,
Here, we investigate a block Rayleigh fading system wherein two uninformed single-antenna transmitters want to transmit a common message to a single-antenna receiver.
%Consider a block Rayleigh fading channel with two separate single-antenna transmitters and one single-antenna receiver.
% (see \cref{fig_two_transmitter}). 
%Assume that both transmitters send the same information to the receiver. The transmitters have no access to their channel state information and there is no link between the transmitters. 
%Also assume that the fading coefficients are constant during two transmission blocks. 
%Therefore, the maximum throughput of the separate transmitter channel is the same as those of a $2 \times 1$ MISO channel with the same total power restriction.  
Let $h_1$ and $h_2$ denote the fading coefficients of the first transmitter-receiver link and second transmitter-receiver link, respectively. 
We assume that $h_1$ and $h_2$ are independent i.i.d.\ complex Gaussian random variables, each with zero-mean and equal variance real and imaginary parts ($h_1, h_2 \sim \mathcal{CN}(0,1)$).
We also assume that $h_1$ and $h_2$ are constant during two consecutive transmission blocks. 
%The fading coefficients of the first transmitter-receiver link and the second transmitter-receiver link, namely $h_1$ and $h_2$, respectively, are independent i.i.d.\ complex Gaussian random variables, each with zero-mean and equal variance real and imaginary parts ($h_1, h_2 \sim \mathcal{CN}(0,1)$) and are constant during two consecutive transmission blocks. 
%Rayleigh-fading network is a symmetric network with $h_{r_i}, h_{i} \sim \mathcal{CN}(0,1)$.

We propose a practical distributed algorithm that provides all instantaneous mutual information distributions which are achievable by treating the transmitters as antennas of one composed element. 
\Cref{distributed-theorem} proves that the outage probability in a MISO channel with two transmit antennas is also achievable in this channel. 

%\begin{figure}
%\centering
%%\epsfig{file=fig3_.eps,width=0.25\linewidth,clip=}
%\includegraphics[scale=0.3]{two_transmitter.eps}
%\caption{A channel with two separate single-antenna transmitters and one single-antenna receiver.}
%\label{fig_two_transmitter}
%\end{figure}

\begin{theorem}  \label{distributed-theorem}
The outage probability in a MISO channel with two transmit antennas and total power constraint $P$ is achievable in a distributed antenna system with two single-antenna transmitters and one single-antenna receiver, where the total power constraint at each transmitter is $\frac{P}{2}$. 
%The maximum expected rate of the channel is given by \cref{miso-sl-rate-max} with $n_t=2$.
\end{theorem}

\begin{proof}
To prove the statement, first, a general expression for the outage probability in a $2 \times 1$ MISO channel is derived. Afterwards, we shall show that this expression is achievable in the two-transmitter distributed antenna system. 

In the $2 \times 1$ MISO channel, the outage probability for transmission rate $R$ is expressed as
\begin{align} \label{miso_single_throughput_not_max}
\mathcal P_{\text{out}} = \Pr \left\{ \ln \left(1+\vec h \mathbf{Q} \vec h^{\dag} \right) < R  \right\},
\end{align}
where $\mathbf{Q}$ is the transmit covariance matrix.
Since $\mathbf Q$ is non-negative definite, one can write it as $\mathbf Q = \mathbf U\mathbf D\mathbf U^\dag$, where $\mathbf D$ is diagonal and $\mathbf U$ is unitary. As $h_1$ and $h_2$ are independent complex Gaussian random variables, each with independent zero-mean and equal variance real and imaginary parts, the distribution of $\vec h\mathbf U$ is the same as that of $\vec h$ \cite{telatar}. Thus, \cref{miso_single_throughput_not_max} is simplified to
\begin{align} \label{miso_single_throughput_not_max_diagonalized}
   \mathcal P_{\text{out}} &= \Pr\left\{\ln \left( 1+\left(\vec h \mathbf U\right)\mathbf D\left(\vec h \mathbf U \right)^\dag \right) < R \right\} \nonumber \\
   &= \Pr\left\{\ln \left( 1+\vec h \mathbf D \vec h^\dag \right) < R \right\}.
\end{align} 
%Assume that $h''_1$ and $h''_2$ are independent complex Gaussian random variables, each with independent zero-mean and equal variance real and imaginary parts. 
Since $\mathbf U_0 = \frac{1}{\sqrt{2}}\left[ \begin{matrix}
1 & 1 \\
1 & -1 \end{matrix} \right]$ is unitary, the distribution of $\vec h \mathbf U_0$ is the same as that of $\vec h$. Inserting into \cref{miso_single_throughput_not_max_diagonalized} yields
\begin{align} \label{miso_single_throughput_not_max_undiagonalized}
   \mathcal P_{\text{out}} 
   &= \Pr\left\{\ln \left( 1+\left( \vec h \mathbf U_0 \right)\mathbf D \left( \vec h \mathbf U_0 \right)^\dag \right) < R \right\} 
\nonumber \\
&= \Pr\left\{\ln \left( 1+\vec h \left( \mathbf U_0 \mathbf D  \mathbf U_0 \right)\vec h^\dag \right) < R \right\}.
\end{align}
Since $\text{tr}\left( \mathbf Q\right)=\text{tr}\left( \mathbf D\right)$, the total power constraint can be written as $\text{tr}\left( \mathbf D\right)\leq P$. Without loss of generality, let us define $\mathbf D \Def P \begin{bmatrix} \delta & 0 \\ 0 & \overline \delta \end{bmatrix}$, where $0 \leq \delta \leq 1$ and $\overline \delta=1- \delta$. Inserting into \cref{miso_single_throughput_not_max_undiagonalized} yields
\begin{align} \label{miso_single_throughput_not_max_undiagonalized2}
   \mathcal P_{\text{out}}  = \Pr\left\{\ln \left( 1+\vec h \frac{P}{2}\left[ \begin{matrix}
1 & 2\delta-1 \\
2\delta-1 & 1 \end{matrix} \right] \vec h^\dag \right) < R \right\}.
\end{align}
Defining $\rho \Def 2\delta-1$, we get
\begin{align} \label{miso_single_throughput_not_max_undiagonalized3}
   \mathcal P_{\text{out}} &{=} \Pr\left\{\ln \left( 1+\vec h \frac{P}{2}\left[ \begin{matrix}
1 & \rho \\
\rho & 1 \end{matrix} \right] \vec h^\dag \right) < R \right\} \nonumber \\
&{=} \Pr\left\{\ln\! \left(\! 1\!+\!  \left( \left|h_1\right|^2\!+\!\left|h_2\right|^2\!+\!2\rho\mathcal \Re(h_1 h_2^*) \right) \frac{P}{2} \right)\! < R \!\right\}.
\end{align}
%where $\hat P \Def \frac{P}{2}$.
Note that as $0 \leq \delta \leq 1$, we have $-1 \leq \rho \leq 1$.

%Also assume that $\rho$ is the correlation coefficient of the transmitted signals of two antennas. The covariance matrix of the transmitted signal is
%\begin{align}
%\mathbf Q = 
%\left[
%\begin{matrix}
%\delta P & \rho \sqrt{\delta \overline{\delta}} P \\
%\rho \sqrt{\delta \overline \delta} P & \overline \delta P
%\end{matrix}
%\right] = 
%\sqrt{\delta \overline{\delta}} P
%\left[
%\begin{matrix}
%\sqrt{\frac{\delta}{\overline{\delta}}} & \rho \\
%\rho  & \sqrt{\frac{\overline \delta}{\delta}}
%\end{matrix}
%\right].
%\end{align}
%Thus, the equivalent channel gain at the receiver is
%\begin{align} \label{miso_2_1_general_rate}
%a_{eq} = \delta |h_1|^2+\overline \delta |h_2|^2+2\rho \sqrt{\delta \overline{\delta}} \Re (h_1 h_2^*).
%\end{align} 

%The expected-rate of this channel for a fixed transmission rate $R=\ln \left( 1+P s \right)$ is
%\begin{align} \label{miso_2_1_general_rate}
%\mathcal R_{s} = \Pr \left\{ \delta |h_1|^2+\overline \delta |h_2|^2+2\rho \sqrt{\delta \overline{\delta}} \Re (h_1 h_2^*) \geq s \right\} \ln \left( 1+P s  \right)
%\nonumber \\
% = \Pr \left\{ a_{eq} \geq s \right\} \ln \left( 1+P s  \right)= \overline F_{a_{eq}}(S)\ln \left( 1+P s  \right).
%\end{align} 

We shall now show that the outage probability in \cref{miso_single_throughput_not_max_undiagonalized3} is achievable in the two-transmitter distributed antenna system with power constraint $\frac{P}{2}$ at each transmitter. 
%This part of the proof consists of three steps: 
%\begin{enumerate}
%\item Encoding procedure
%\item Decoding procedure
%\item Expected-rate evaluation
%\end{enumerate} 
%\begin{description}
%\item{1-} Encoding procedure
%\item{2-} Decoding procedure
%\item{3-} expected-rate evaluation
%\end{description} 
%%%%%%%%%%%%%%%%%%%%%%%%%%%%%%

%\subsection{Encoding Procedure}
%Without loss of generality assume that $\delta > \overline \delta$, because the roles of $\delta$ and $\overline \delta$ are interchangeable and the proof is symmetric in the two variables. 
The transmission strategy in two consecutive time slots is as follows. 
In time slot $t$, the first (resp. second) transmitter sends $X(t)$ (resp. $\rho X(t) + \sqrt{\left(1-\rho^2\right)} X(t+1)$). 
In time slot $t+1$, the first (resp. second) transmitter sends $-X^*(t+1)$ (resp. $-\rho  X^*(t+1) + \sqrt{\left(1-\rho^2\right)} X^*(t)$).
Assuming $\mathbb E\left( \left|X \right|^2\right)=\frac{P}{2}$, the power consumption per time slot in each transmitter is $\frac{P}{2}$.
%Note that since $\delta \geq \overline \delta$, then $\delta + \rho^2 \overline \delta - \frac{1}{2}>0$. 

The received signal at the receiver is
\begin{align} \label{alam_2trans_1}
&Y(t)=h_{1} X(t)+h_{2}\Big( \rho X(t)  \nonumber \\
&\qquad + \sqrt{\left(1-\rho^2\right)} X(t+1) \Big)+Z(t), \\
%\end{align}
%\begin{align}
&Y(t\!+\!1)\!=\!-h_{1}\! X^*(t\!+\!1) \!+\!h_{2}\! \Big(\!\! -\rho X^*(t+1)  \nonumber \\
&\qquad \ \ \ \ + \sqrt{\left(1-\rho^2\right)} X^*(t) \Big)+Z(t+1). 
\end{align}
In matrix form,
\begin{align} \label{alam_2trans_2}
 \begin{bmatrix}
  Y(t) \\
-Y(t+1)^{*}
\end{bmatrix}
=
\mathbf G 
\begin{bmatrix}
  X(t) \\
X(t+1)
\end{bmatrix}
+
 \begin{bmatrix}
  Z(t) \\
-Z^{*}(t+1)
\end{bmatrix},
\end{align}
where
\begin{align}
\mathbf G \!\! \Def \!\!\! \begin{bmatrix}
 h_{1} \!+\! h_{2}\rho &  h_{2}\sqrt{\left(1-\rho^2\right)}  \\
 -h_{2}^{*}\sqrt{\left(1-\rho^2\right)} &  h_{1}^{*}\!+\!h_2^* \rho 
 \end{bmatrix}\!\!.
\end{align}

%%%%%%%%%%%%%%%%%%%%%%%%%%%%%%

%\subsection{Decoding Procedure} 

By multiplying $\mathbf G^\dag$ to the both sides of \cref{alam_2trans_2}, two parallel channels are separated as
\begin{align} \label{alamouti_decoder_matrix}
 \begin{bmatrix}
  \tilde{Y}(t) \\
\tilde{Y}(t+1)
 \end{bmatrix}
&= 
\mathbf G^\dag
\begin{bmatrix}
  Y(t) \\
-Y^{*}(t+1)
 \end{bmatrix} 
 \nonumber \\
& =\Big( \left|h_1 +h_2\rho  \right|^{2} \nonumber \\
&+|h_{2}|^{2}\left(1-\rho^2\right)\Big)
%\begin{bmatrix}
% 1 & 0 \\
%0 & 1
%\end{bmatrix}
\mathbf I_2
\begin{bmatrix}
  X(t) \\
X(t+1)
 \end{bmatrix} 
 \nonumber \\
&+
\mathbf G^\dag
\begin{bmatrix}
  Z(t) \\
-Z^{*}(t+1)
 \end{bmatrix} \nonumber \\
 &=
  h \mathbf I_2
\begin{bmatrix}
  X(t) \\
X(t+1)
 \end{bmatrix} 
+
\begin{bmatrix}
  \tilde Z(t) \\
 \tilde Z(t+1)
 \end{bmatrix},
\end{align}
where $h \Def \left|h_1 +h_2\rho \right|^{2} +\left|h_{2}\right|^{2}\left(1-\rho^2\right)$, and $\tilde Z(t)$ and $\tilde Z(t+1)$ are independent zero mean complex Gaussian random variables with power equal to $\mathbb E \left(\left|\tilde Z\right|^2\right)=h$. Thus, the received signal power to noise ratio at the receiver is 
\begin{align} \label{2transmitter_general_rate}
\frac{h^2 \frac{P}{2}}{\mathbb E \left(\left|\tilde Z\right|^2\right)} &= \left(\left|h_1 +h_2\rho \right|^{2} +\left|h_{2}\right|^{2}\left(1-\rho^2\right)\right)\frac{P}{2}\nonumber \\
&= \left(\left|h_1\right|^2+\left|h_2\right|^2+2\rho\Re\left(h_1 h_2^*\right)\right) \frac{P}{2}.
\end{align}
Therefore, the outage probability in the proposed scheme is given by
\begin{align} \label{miso_single_throughput_not_max_undiagonalized4}
   \mathcal P_{\text{out}} {=} \Pr\left\{\ln\! \left( \!1\!+\!  \left( \left|h_1\right|^2\!+\!\left|h_2\right|^2\!+\!2\rho\mathcal \Re(h_1 h_2^*) \right) \frac{P}{2}\right)\! < \!R \right\}
.
\end{align}
\Cref{miso_single_throughput_not_max_undiagonalized3} together with \cref{miso_single_throughput_not_max_undiagonalized4} shows that the outage probability in a $2\times 1$ MISO channel is also achievable in the two-transmitter distributed antenna system.

\end{proof}

\begin{remark}
To achieve the minimum outage probability in \cref{distributed-theorem}, the optimum solution to $\delta$ is either $1$ or $\frac{1}{2}$, depending on $R$ and $P$. Equivalently, in the two-transmitter distributed antennas, the optimum value of $\rho$ is either $1$ or $0$.
\end{remark}

%\begin{remark}
Note that for $\rho=0$, the proposed transmission scheme in the two-transmitter distributed antenna system is equivalent to the Alamouti code \cite{alamouti}. 
%When $\rho=1$, the proposed scheme is equivalent to transmitting the same signals from both transmitters.
%\end{remark}

\begin{remark}
Since the outage probability is the CDF of the instantaneous mutual information, one concludes that any achievable instantaneous mutual information distribution in the $2\times 1$ MISO channel is also achievable in this two-transmitter distributed antenna system. 
\end{remark}

\begin{remark}
Based on \cref{distributed-theorem}, the maximum throughput in the two-transmitter distributed antenna system with total power constraint $\frac{P}{2}$ at each transmitter is the same as that of a $2\times 1$ MISO channel with total power constraint $P$. By substituting $n_t=2$ in \cref{miso-sl-rate-max}, the maximum throughput is given by
\begin{align}
\mathcal R_s^m = \max_{0< s< 1} (1+2s)e^{-2s} \ln \left( 1+Ps \right).
\end{align}
\end{remark}

\begin{remark}
In a similar approach, it can be shown that the maximum expected-rate as well as the ergodic capacity of this two-transmitter distributed antenna system and the $2\times 1$ MISO channel are the same.
\end{remark}

%\begin{remark}
%In many scenarios of MISO channels, the optimal transmit covariance matrix is neither diagonal nor equal power. One example of such scenarios is the channels with side information obtained via perfect or imperfect feedback or other means, e.g., the channels with non-zero mean vector or non-unity covariance matrix \cite{narula,visotsky}. The non-cooperative distributed antenna counterpart of these scenarios is an application of \cref{distributed-theorem}.
%\end{remark}

%\begin{remark}
%Simple cases:
%\begin{itemize}
%\item $\delta = \frac{1}{2}$ and $\rho = 0$: the proposed scheme simplifies to Alamouti code \cite{alamouti}
%\item $\delta = 1$: the proposed scheme simplifies to sending the same signal from both transmitters
%\item $\delta = \frac{1}{2}$ and $\rho = 1$: the proposed scheme simplifies to sending the same signal from both transmitters
%\end{itemize} 
%\end{remark}

%\Cref{distributed-proposition} generalizes \cref{distributed-theorem} to multi-layer codes.
Based on \cref{distributed-theorem} and recall from \cref{miso-infinite-proposition} with $n_t = 2$, we come up with the following Corollary. 
\begin{corollary} \label{distributed-proposition}
%Based on \cref{distributed-theorem}, every achievable instantaneous mutual information in a $2 \times 1$ MISO channel, is also achievable by two non-cooperative distributed antennas, and clearly, vice versa. Hence, these channels can be called equivalent in the sense of mutual information, and as a result, their maximum throughput and expected-rate, as well as their ergodic capacity and outage capacity are the same. 
The maximum continuous-layer expected-rate of the distributed antenna system with two transmitters each with total power $\frac{P}{2}$ is
\begin{align}
\mathcal R_{c}^m = 3\emph{E}_1(s_0) + (1-s_0)e^{-s_0} - 3\emph{E}_1(s_1) - (1-s_1)e^{-s_1},
\end{align}
where $s_1 = \frac{1+\sqrt{5}}{2}$, and $s_0 = \sqrt[3]{\sqrt{ A^2-B^3}+A}  
 +\frac{B}{\sqrt[3]{\sqrt{ A^2-B^3}+A}}-\frac{2}{3 P}$ with $A=\frac{1}{ P}-\frac{2}{3 P^2}-\frac{8}{27 P^3}$ and $B=\frac{2}{3 P}+\frac{4}{9 P^2}$.
 
From \cref{ergodic-capacity-miso-pro}, the ergodic capacity in this channel is 
\begin{align}
C_{\emph{erg}} = 1+\left(1- \frac{2}{P} \right) e^{\frac{2}{P}} \emph{E}_1\left( \frac{2}{P} \right)  .
\end{align}

\end{corollary}
%\begin{proof}
%The proof follows directly from \cref{distributed-theorem} and \cref{miso-infinite-proposition}. 
%\end{proof}

The maximum throughput, the maximum two-layer expected-rate, the maximum continuous-layer expected-rate, and the ergodic capacity in the two-transmitter distributed antenna system are depicted in \cref{miso-two-antenna-fig}.
% with $n_t=2$.

\begin{figure}
\centering
\includegraphics[scale=0.7]{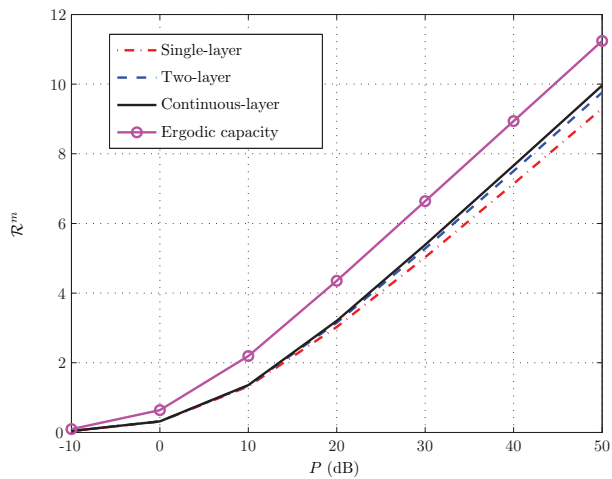}
\caption{The maximum throughput, the maximum two-layer expected-rate, the maximum continuous-layer expected-rate, and the ergodic capacity (all in \emph{nats}) in the two-transmitter distributed antenna system.}
\label{miso-two-antenna-fig}
\end{figure}

%%%%%%%%%%%%%%%%%%%%%%%%%%%%%%%%%%%%%%%%%%%%%%%%%%%%%%%%%%%%%%%%%%%%%%%%%%%%%%%%%%%%%%%%%%%%%%%%%%%%%%%%%%%%%%%%%%%
%%%%%%%%%%%%%%%%%%%%%%%%%%%%%%%%%%%%%%%%%%%%%%%%%%%%%%%%%%%%%%%%%%%%%%%%%%%%%%%%%%%%%%%%%%%%%%%%%%%%%%%%%%%%%%%%%%%
%%%%%%%%%%%%%%%%%%%%%%%%%%%%%%%%%%%%%%%%%%%%%%%%%%%%%%%%%%%%%%%%%%%%%%%%%%%%%%%%%%%%%%%%%%%%%%%%%%%%%%%%%%%%%%%%%%%
%%%%%%%%%%%%%%%%%%%%%%%%%%%%%%%%%%%%%%%%%%%%%%%%%%%%%%%%%%%%%%%%%%%%%%%%%%%%%%%%%%%%%%%%%%%%%%%%%%%%%%%%%%%%%%%%%%%
%%%%%%%%%%%%%%%%%%%%%%%%%%%%%%%%%%%%%%%%%%%%%%%%%%%%%%%%%%%%%%%%%%%%%%%%%%%%%%%%%%%%%%%%%%%%%%%%%%%%%%%%%%%%%%%%%%%

\section{Conclusion} \label{conclusion-section}

The throughput and expected-rate maximization of multiple-antenna channels are addressed in block Rayleigh fading environments, in which the transmitter does not access the CSI. 
%It is established that in terms of the impact of the transmit covariance matrix, the behavior of the maximum throughput is similar to the behavior of the ergodic capacity which is to send uncorrelated circularly symmetric zero mean equal power Gaussian signals from all the transmit antennas.
It is established that, in order to achieve the maximum throughput, one has to transmit uncorrelated circularly symmetric zero mean equal power Gaussian signals from all the transmit antennas. This indeed yields the same transmit covariance matrix that achieves the ergodic capacity.

In point-to-point uncorrelated MISO channels, in contrast to using a fraction of antennas which is optimum for outage capacity, the throughput is maximized by sending uncorrelated equal power signals on all transmit antennas. 
%Afterwards, the maximum expected-rate is analyzed by sending multi-layer codes at the transmitter. 
The maximum expected-rate is analyzed using multi-layer codes. 
It is proved that in each layer, sending uncorrelated signals with equal powers from all available antennas is optimum. The continuous-layer expected-rate of the channel is then derived in closed form.
%It is also shown that channel correlation decreases the maximum throughput as well as the maximum expected-rate.

The optimum transmit strategy maximizing the throughput is obtained for point-to-point uncorrelated MIMO channels. Since the PDF of the MIMO instantaneous mutual information is not tractable, four asymptotic cases are considered: low SNR regime, high SNR regime, large number of transmit antennas, and large number of receive antennas.
In each case, the maximum throughput of the MIMO channel is derived.

Finally, a distributed antenna system with two single-antenna transmitters and one single-antenna receiver is investigated. It is proved that any achievable instantaneous mutual information distribution in the $2\times 1$ MISO channel is also achievable in the two-transmitter distributed antenna system. 
Hence, both systems achieve the same maximum throughput and expected-rate.

\begin{appendices} 

\section{Proof of \Cref{ergodic-capacity-miso-pro}} \label{appendix-ergodic-capacity-miso-pro}

%From \cite{telatar}, 
The ergodic capacity of a $1 \times n_r $ SIMO channel is given by
\begin{align} \label{ergodic-capacity-general-eq}
C_{\text{erg}} = \int_{0}^{\infty} \frac{x^{n_r-1}e^{-x}}{\left( n_r-1 \right)!} \ln \left( 1+P x \right) \ud x.
\end{align}
Applying the integration by parts rule on \cref{ergodic-capacity-general-eq} leads to
\begin{align} \label{ergodic-capacity-general-eq2}
C_{\text{erg}} &= \left[-e^{-x} \sum_{\ell=0}^{n_r-1} \frac{x^\ell}{\ell!} \ln\left( 1+P x \right) \right]_{0}^{\infty} \nonumber \\
&+ \int_{0}^{\infty} e^{-x} \sum_{\ell=0}^{n_r-1} \frac{x^\ell}{\ell!} \frac{P}{1+Px} \ud x.
\end{align}
One can simply show that the first part on the right-hand-side in \cref{ergodic-capacity-general-eq2} is zero by repeatedly applying l'H$\hat{\text{o}}$pital's rule.
With $t = 1+P x$, \cref{ergodic-capacity-general-eq2} yields
\begin{align} \label{ergodic-capacity-general-eq3}
C_{\text{erg}} &= \int_{1}^{\infty} e^{-\frac{t-1}{P}} \sum_{\ell=0}^{n_r-1} \frac{1}{t \ell!} \left( \frac{t-1}{P} \right)^\ell \ud t.
%\nonumber \\
%&= e^{\frac{1}{P}} \int_{1}^{\infty} .
\end{align}
From $\left( t-1 \right)^\ell=\sum_{\imath=0}^\ell \binom{\ell}{\imath} t^\imath \left(-1\right)^{\ell-\imath}$, where $\binom{\ell}{\imath}$ is the binomial coefficient, we get
\begin{align} \label{ergodic-capacity-general-eq4}
C_{\text{erg}}&=e^{\frac{1}{P}} \int_{1}^{\infty} e^{-\frac{t}{P}} \sum_{\ell=0}^{n_r-1} \frac{1}{P^\ell \ell! t} \sum_{\imath=0}^{\ell} \binom{\ell}{\imath} t^\imath \left(-1\right)^{\ell-\imath} \ud t \nonumber \\
&= e^{\frac{1}{P}} \sum_{\ell=0}^{n_r-1} \frac{\left(-1\right)^{\ell}}{P^\ell \ell!} \int_{1}^{\infty} \frac{e^{-\frac{t}{P}}}{t} \ud t \nonumber \\
&+ e^{\frac{1}{P}} \sum_{\ell=1}^{n_r-1} \frac{1}{P^\ell \ell!} \sum_{\imath=1}^{\ell} \left(-1\right)^{\ell-\imath} \binom{\ell}{\imath} \int_{1}^{\infty} e^{-\frac{t}{P}} t^{\imath-1} \ud t .
\end{align}
With $u=\frac{t}{P}$, we have 
\begin{align} \label{ergodic-capacity-general-eq5}
\int_{1}^{\infty} e^{-\frac{t}{P}} t^{\imath-1} \ud t  &= P^\imath \int_{\frac{1}{P}}^{\infty} e^{-u} u^{\imath-1} \ud u \nonumber \\
&= \left(\imath-1\right)! P^\imath e^{-\frac{1}{P}} \sum_{m=0}^{\imath-1} \frac{1}{m!} \left( \frac{1}{P} \right)^m.
\end{align}
Inserting \cref{ergodic-capacity-general-eq5} into \cref{ergodic-capacity-general-eq4}, we obtain
\begin{align} \label{ergodic-capacity-general-eq6}
C_{\text{erg}}&= e^{\frac{1}{P}} \text{E}_1\left({\frac{1}{P}}\right) \sum_{\ell=0}^{n_r-1} \frac{\left(-1\right)^{\ell}}{P^\ell \ell!} \nonumber \\
&+ \sum_{\ell=1}^{n_r-1} \frac{1}{P^\ell} \sum_{\imath=1}^{\ell} \frac{\left(-1\right)^{\ell-\imath}}{\imath\left(\ell-\imath\right)!}P^\imath \sum_{m=0}^{\imath-1} \frac{1}{m!} \frac{1}{P^m}.
\end{align}
Let $k=\ell-\imath$, the above leads to
\begin{align} \label{ergodic-capacity-general-eq7}
C_{\text{erg}}&= e^{\frac{1}{P}} \text{E}_1\left({\frac{1}{P}}\right) \sum_{\ell=0}^{n_r-1} \frac{\left(-1\right)^{\ell}}{P^\ell \ell!} \nonumber \\
&+ \sum_{\ell=1}^{n_r-1} \sum_{k=0}^{\ell-1} \frac{\left(-1\right)^{k}}{\left(\ell-k\right)k!} \sum_{m=0}^{\ell-k-1} \frac{1}{m! P^{k+m}}.
\end{align}

From \cite{telatar}, the ergodic capacity in an $n_t \times 1$ MISO channel with total power constraint $P$ equals the ergodic capacity in a $1 \times n_t$ SIMO channel with total power constraint $\frac{P}{n_t}$. Hence, we obtain \cref{ergodic-capacity-miso-formula} by replacing $P$ with $\frac{P}{n_t}$ and $n_r$ with $n_t$ in \cref{ergodic-capacity-general-eq7}.

\section{} \label{proof-expected-rate-integral}

The indefinite integral (antiderivative) of \cref{infinite-layer-miso-integral} can be written as
\begin{align} \label{infinite-layer-miso-integral-appendix}
\mathcal R(s) &= \int e^{-s} \left( \frac{n_t+1}{s} -1  \right) \sum_{\ell =0}^{n_t-1} \frac{s^\ell}{\ell !} \ud s \nonumber \\
&=  \left(n_t+1\right) \int \frac{e^{-s}}{s} \ud s
+ \left(n_t+1\right) \int e^{-s} \sum_{\ell =0}^{n_t-1} \frac{s^{\ell-1}}{\ell !} \ud s \nonumber \\
&- \int e^{-s} \sum_{\ell =0}^{n_t-1} \frac{s^{\ell}}{\ell !} \ud s \nonumber \\
&= \left(n_t+1\right)\! \int \frac{e^{-s}}{s} \ud s 
+\! \sum_{\ell =0}^{n_t-1} \frac{1}{\ell!} \Bigg(\!\left(n_t+1\right) \! \int \! s^{\ell-1} e^{-s} \ud s \nonumber \\
&- \int s^{\ell} e^{-s}  \ud s \Bigg).
\end{align}
The definite integral of $\mathcal R(s)$ over the interval $[s_0 ~~ \infty]$ is given by
\begin{align} \label{infinite-layer-miso-integral-appendix2}
\left[\mathcal R(s)\right]_{s_0}^{\infty} &= \left(n_t+1\right) \int_{s_0}^{\infty} \frac{e^{-s}}{s} \ud s
+ \sum_{\ell =0}^{n_t-1} \frac{1}{\ell!} \Bigg( \nonumber \\
&\left(n_t+1\right) \int_{s_0}^{\infty} s^{\ell-1} e^{-s} \ud s 
- \int_{s_0}^{\infty} s^{\ell} e^{-s}  \ud s \Bigg) \nonumber \\
&= \left(n_t+1\right) \text{E}_1 \left(s_0\right) + \sum_{\ell =0}^{n_t-1} \frac{1}{\ell!} \Bigg( \nonumber \\
& \left(n_t+1\right) \left(\ell-1\right)! e^{-s_0} \sum_{k=0}^{\ell-1} \frac{s_0^k}{k!}
- \ell! e^{-s_0} \sum_{k=0}^{\ell} \frac{s_0^k}{k!} \Bigg) \nonumber \\
&= \left(n_t+1\right) \text{E}_1 \left(s_0\right) -e^{-s_0} +e^{-s_0} \sum_{\ell =1}^{n_t-1} \frac{1}{\ell!} \Bigg( \nonumber \\
&- s_0^\ell + \left(n_t+1-\ell\right) \left(\ell-1\right)! \sum_{k=0}^{\ell-1} \frac{s_0^k}{k!} \Bigg).
\end{align}
The definite integral of $\mathcal R(s)$ over the interval $[s_0 ~~ s_1]$ can be written as $\left[\mathcal R(s)\right]_{s_0}^{\infty}-\left[\mathcal R(s)\right]_{s_1}^{\infty}$. Therefore, defining
\begin{align} \label{infinite-layer-miso-integral-appendix3}
\mathcal R(s) &\Def -\left(n_t+1\right) \text{E}_1 \left(s\right) +e^{-s} \nonumber \\
&+e^{-s} \sum_{\ell =1}^{n_t-1} \frac{1}{\ell!} \Bigg( s^\ell - \left(n_t+1-\ell\right) \left(\ell-1\right)! \sum_{k=0}^{\ell-1} \frac{s^k}{k!} \Bigg), 
\end{align}
and inserting into \cref{infinite-layer-miso-integral-appendix2} leads to the conclusion that
\begin{align}
\left[\mathcal R(s)\right]_{s_0}^{s_1} = \mathcal R(s_1) - \mathcal R(s_0).
\end{align}

\end{appendices}

\bibliographystyle{IEEEtran}
\bibliography{bibliography}
\end{document}